\documentclass{article}

\usepackage{preamble}
\usepackage{comment}
\usepackage[letterpaper,margin=1.0in]{geometry}

\title{Quantum Algorithms for Projection-Free Sparse Convex Optimization}

\author{Jianhao He\textsuperscript{ 1}\qquad John C.S. Lui\textsuperscript{ 1}\\
\vspace{-2mm} \\
\normalsize{\textsuperscript{1} The Chinese University of Hong Kong
}\vspace{-2mm}\\
\\
\normalsize{\texttt{jianhaohe9@cuhk.edu.hk\qquad cslui@cse.cuhk.edu.hk}}}

\usepackage{color, colortbl}
\usepackage{xcolor}

\usepackage{threeparttable, array, float} 

\definecolor{red}{HTML}{E51400}  
\definecolor{blue}{HTML}{0050EF} 
\definecolor{green}{HTML}{008A00} 
\definecolor{purple}{HTML}{AA00FF} 

\date{}

\begin{document}

\maketitle

\begin{abstract}
    This paper considers the projection-free sparse convex optimization problem for the vector domain and the matrix domain, which covers a large number of important applications in machine learning and data science. 
    For the vector domain $\cD \subset \R^d$, we propose two quantum algorithms for sparse constraints that finds a $\varepsilon$-optimal solution with the query complexity of $O(\sqrt{d}/\varepsilon)$ and $O(1/\varepsilon)$ by using the function value oracle, reducing a factor of $O(\sqrt{d})$ and $O(d)$ over the best classical algorithm, respectively, where $d$ is the dimension.
    For the matrix domain $\cD \subset \R^{d\times d}$, we propose two quantum algorithms for nuclear norm constraints that improve the time complexity to $\tilde{O}(rd/\varepsilon^2)$ and $\tilde{O}(\sqrt{r}d/\varepsilon^3)$ for computing the update step, reducing at least a factor of $O(\sqrt{d})$ over the best classical algorithm, where $r$ is the rank of the gradient matrix.
    Our algorithms show quantum advantages in projection-free sparse convex optimization problems as they outperform the optimal classical methods in dependence on the dimension $d$.
\end{abstract}

\section{Introduction}

In this paper, we consider the following \textit{constrained} optimization problem of the form 
\begin{equation}\label{eq:general_constraint_convex}
    \min_{\bx \in \cD}f(\bx),
\end{equation}
such objective covers many important application in operations research and machine learning.
We are interested in the case where 1) the objective function $f$ is convex and continuously differentiable, and 2) the domain $\cD \subset \R^d$ is a feasible set that is convex, and the dimension $d$ is high. 
Typical instances of such high-dimensional optimization problems include multiclass classification, multitask learning, matrix learning, network systems and many more \cite{garber2016linearly,hazan2012projection,hazan2012near,jaggi2013revisiting,dudik2012lifted,zhang2012accelerated,harchaoui2015conditional,hazan2016variance}.
As an example, for the matrix completion, the optimization problem to be solved is as follows:

\begin{align}\label{eq: matrix_completion_example}
    \min_{X \in \R^{m \times n}, \norm{X}_{\tr} \le r}\sum_{(i,j)\in \Omega}(X_{i,j}-Y_{i,j})^2,
\end{align}
where $X$ is the matrix to be recovered, $\Omega$ denotes the observed elements, $Y_{i,j}$ is the observed known value at position $(i,j)$, and $\norm{X}_{\tr}\le r$ represents the trace norm (nuclear norm) constraint.

Compared with unconstrained convex optimization problems, optimizing \cref{eq:general_constraint_convex} involves handling constraints, which introduces new challenges. 
A straightforward method for optimizing \cref{eq:general_constraint_convex} is the projected gradient descent approach \cite{levitin1966constrained}. 
This method first takes a step in the gradient direction and then performs the projection to satisfy the constraint. 
However, in practice, the dimensions of the feasible set can be very large, leading to prohibitively high computational complexity.
For example, when solving \cref{eq: matrix_completion_example}, the projection step involves performing a singular value decomposition (SVD), whose time complexity is $O(mn \min\{m,n\})$ ($O(d^3)$ for $X\in \R^{d \times d}$). 
Compared to the projected gradient descent approach, the Frank-Wolfe (FW) method (also known as the conditional gradient method) is more efficient when dealing with structured constrainted optimization problems.
Rather than performing projections, it solves a computationally efficient linear sub-problem to ensure that the solution lies within the feasible set $\cD$. 
When solving \cref{eq: matrix_completion_example}, the time complexity of the Frank-Wolfe method is $O(mn)$  ($O(d^2)$ for $X\in \R^{d \times d}$), which is significantly lower than the complexity of SVD-based projections. Since the Frank-Wolfe method is efficient for optimizing many difficult machine learning problems, such as low-rank constrained problems and sparsity-constrained problems, it has attracted significant attention and has been applied to solving \cref{eq:opt1} and many of its variants.

Despite the efficiency of FW in handling structured constraints, it still incurs significant computational overhead when dealing with high-dimensional problems. 
The bottleneck of the computation is the linear subproblem over $\cD$, which is either assumed to have efficient implementation or simply follows existing classical oracles, such as \cite{dunn1978conditional,jaggi2013revisiting,garber2016linearly}. 
The overhead of these oracles, however, grows linearly or superlinearly in terms of dimension $d$. 

Recently, quantum computing has emerged as a promising new paradigm to accelerate a large number of important optimization problems, e.g.,  
combinatorial optimization \cite{Grover1996,Ambainis2006,Durr2006,durr1996quantum,Mizel2009,Yoder2014,Sadowski2015,He2020}, 
linear programming \cite{KP18,Li2019,Van2019,apers2023quantum},
second-order cone programming \cite{Kerenidis2019s,Kerenidis2019Svm,Kerenidis2019P}, 
quadratic programming \cite{kerenidis2020quantum}, 
polynomial optimization \cite{Rebentrost2019},
semi-definite optimization \cite{KP18,Joran19,Fernando17,Fernando19,Joran17}, 
convex optimization \cite{van2020convex, chakrabarti2020quantum,zhang2024quantum}, 
nonconvex optimization \cite{zhang2023quantum,chen2025quantum}, 
stochastic optimization \cite{sidford2023quantum}
online optimization \cite{he2022quantum,he2024quantum,lim2022quantum}, 
multi-arm bandit \cite{casale2020,wang2020,li2022quantum,wan2023quantum}.
We aim to take a thorough investigation on whether quantum computing can accelerate FW algorithms, and in particular the linear sub-problem over structured constraints regarding dimension $d$. We aim to answer the following question:

\textit{Can one utilize quantum techniques to accelerate Frank-Wolfe algorithms in terms of dimension $d$?}

Chen et al. gave an initial answer to this question \cite{chen2021quantum}. 
They considered the linear regression problem with explicit functional form where the closed form of gradient is provided.
Given the precomputed matrix factors of the closed-form objective function stored in specific data structures, they
leveraged HHL-based algorithms to accelerate matrix multiplications in calculating the closed-form gradient, leading to a upper bound of $O\left(\sqrt{d}/\varepsilon^2\right)$.
In this work, we consider a more general problem where the objective function is a smooth convex function accessible only through a function value oracle, and then we consider a more general constraint conditions (\textit{the latent group norm ball})  to enhance the theoretical framework's applicability.
Besides, we also consider the case of matrix feasible set, under different assumptions. 
To our best knowledge, we are the first one to consider accelerating the matrix case of the FW algorithm by quantum computing.

\textbf{Contributions.} We give a systematic study on how to accelerate FW algorithms when $\cD$ is either a vector domain $\R^d$, or a matrix domain $\R^{d \times d}$ subject to various structured constraints.
Note that our findings can be applied to non-square matrices, we express our results using square matrices for simplicity of presentation (Remark \ref{remark:square}).
We summarize our contributions as follows.

For the vector domain $\cD \subset \R^d$:

\begin{itemize}
    \item We propose the quantum Frank-Wolfe algorithm for the projection-free sparse convex optimization problem under $\ell_1$ norm constraints (Theorem \ref{Theorem:QFWSC}) and the $d$-dimensional simplex $\Delta_d$ (Theorem \ref{Theorem:QFWS}). We achieve a query complexity of $\widetilde{O}(\sqrt{d}/\varepsilon)$ in finding an $\varepsilon$-optimal solution using the function value oracle, reducing a factor of $O(\sqrt{d})$ over the optimal classical algorithm. Furthermore, if the objective function is a Lipschitz continuous function, we prove that the query complexity can be reduced to $O(1/\epsilon)$ by employing the bounded-error Jordan quantum gradient estimation algorithm, at the cost of more qubits and additional gates (Theorem \ref{Theorem:QFWJ}). In addition, we consider the generalization to latent group norm constraints (Theorem \ref{Theorem:QFWLG}) and achieve a query complexity of $\widetilde{O}\left(\sqrt{\abs{\+G}}\abs{\*g}_{\max}\right)$, representing an $O\left(\sqrt{\abs{\+G}}\right)$ speedup over the classical algorithm. These results are presented in Section \ref{sec:QFW}, Appendix \ref{sec:QFWJ} and \ref{sec:extension}. The comparison with the classical methods is shown in Table \ref{table:resultV}.
    \item Specifically, we construct efficient quantum subroutines to find the maximum gradient component for different constraints and show that it is robust to errors: even if we fail to select the greatest gradient direction, as long as the gradient component of the chosen direction has a bounded error relative to the largest gradient component, convergence can still be ensured by selecting appropriate parameters. 
\end{itemize}

For the matrix domain $\cD \subset \R^{d\times d}$:

\begin{itemize}
    \item We propose two quantum Frank-Wolfe algorithms for the projection-free sparse problem under nuclear norm constraints.
    For finding an $\varepsilon$-optimal solution, we achieve a time complexity of $\tilde{O}(rd/\varepsilon^2)$ (Theorem \ref{theorem:qfwqsvd}) and $\tilde{O}(\sqrt{r}d/\varepsilon^3)$ (Theorem \ref{theorem:qfwqpm}) in computing the update direction, representing an at least $O(\sqrt{d})$ speedup over state-of-the-art classical algorithm, where $r$ is the rank of the gradient matrix.
    Notably, this analysis focuses on the update direction computation and assumes that the gradient has been pre-computed and stored in the memory (Remark \ref{remark:noG}), following the classical convention of excluding gradient evaluation time \cite{jaggi2013revisiting}.
    These results are presented in Section \ref{Sec:QFWM} and the comparison with the classical methods is shown in Table \ref{table:resultM}.
    \item Specifically, in the first algorithm, we simplify the top-$k$ singular vectors extraction method \cite{bellante2022quantum} by utilizing the quantum maximum finding algorithm, which avoids the overheads of repeated sampling to estimate the factor score ratio, and avoids the overheads of searching the threshold value.
    In the second algorithm, we introduce the quantum power method to extract the top singular vectors, which reduces the dependence on the rank of the gradient matrix, at the cost of higher sensitivity on solution precision.
\end{itemize}

We notice an independent simultaneous work on the quantum power method \cite{chen2025quantumQPM}, whose second algorithm shares a conceptual similarity with our second approach: both iteratively apply quantum matrix-vector multiplication. 
However, there are several technical distinctions between the algorithmic frameworks. 
For instance: \cite{chen2025quantumQPM} requires an additional sparse-query access to the matrix as input; 
the complexity dependence in \cite{chen2025quantumQPM} involves the eigenvalue gap, whereas our method relies on some distinct factors; 
the error forms of the outputs are different. 
We believe that the second algorithm in \cite{chen2025quantumQPM} can be adapted to accelerate the Frank-Wolfe algorithm through appropriate modifications, though we reserve this extension for future investigation.

The remainder of this paper is organized as follows.
Section \ref{sec:pre} introduces the basic concept of constrained optimization and the classical Frank-Wolfe algorithm. Appendix \ref{sec:quantumBase} introduces the notations and assumptions of quantum computing.
Section \ref{sec:QFW} and \ref{Sec:QFWM} presents our quantum FW methods for vector domain and matrix domain, respectively.
Extension for the vector cases are presented in Appendix \ref{sec:QFWJ} and \ref{sec:extension}.
Extended related works are presented in Appendix \ref{sec:relatedWorks}, and we conclude with a discussion about the future work in Section \ref{sec:conclude}.
Proof details are given in Appendix \ref{lab:SecProof}.

\begin{table*}[t]
        \setlength{\abovecaptionskip}{0.1cm}
        \setlength{\belowcaptionskip}{-0.1cm}
	\centering
    \caption{Classical algorithms V.S. quantum algorithms of the vector case, where $C_f$ is the curvature of the objective function $f$, $\varepsilon$ is the precision of the solution, $d$ is the dimension of the domain, $G$ is the Lipschitz parameter of the objective function, $p$ is the failure probability.}	
	\resizebox{\linewidth}{!}{
	\centering
	\begin{threeparttable}
	\begin{tabular}{|cc|c|c|c|c|c|c|}
 \hline
		\textbf{Optimization Domain}& \textbf{Constraints}&\textbf{Algorithm}& \textbf{Iteration} & \textbf{Query complexity} & \textbf{Qubits} & \textbf{Gates}\\
\hline
     Sparse Vectors   & $\norm{\cdot}_1$-ball &  FW \cite{jaggi2013revisiting}    & $O(C_f/\varepsilon)$ & $O(d)$ &  & \\

          && \textbf{QFW (Theorem \ref{Theorem:QFWSC})} & $O(C_f/\varepsilon)$ & $O(\sqrt{d} \log{(C_f/p\varepsilon)} )$ & $O\left(d+\log{\frac{1}{\varepsilon}}\right)$ & $O(\sqrt{d})$\\

          && \textbf{QFW (Theorem \ref{Theorem:QFWJ})} & $O(C_f/\varepsilon)$ & $O(1)$ & $O\left(d \log{\frac{Gd}{\rho\varepsilon}}\right)$ & $O(d\log{d})$\\
          \hline
     Sparse non-neg. vectors & Simplex $\Delta_d$  &   Frank-Wolfe \cite{jaggi2013revisiting} & $O(C_f/\varepsilon)$ & $O(d)$ &  &\\
          & & \textbf{QFW (Theorem \ref{Theorem:QFWS})}  &   $O(C_f/\varepsilon)$ & $O(\sqrt{d} \log{(C_f/p\varepsilon)})$ & $O\left(d+\log{\frac{1}{\varepsilon}}\right)$ & $O(\sqrt{d})$\\
    
          && \textbf{QFW (Theorem \ref{Theorem:QFWJ})} & $O(C_f/\varepsilon)$ & $O(1)$ & $O\left(d \log{\frac{C_fGd}{p\varepsilon}}\right)$ & $O(d\log{d})$\\
      
           \hline
      Latent group sparse vectors    &   $\norm{\cdot}_{\cG}$-ball & FW \cite{jaggi2013revisiting} & $O(C_f/\varepsilon)$ & $O(\sum_{g \in \cG}|g|)$ & &\\
      &   & \textbf{QFW (Theorem \ref{Theorem:QFWLG})} & $O(C_f/\varepsilon)$ & $O\left(\sqrt{\abs{\+G}}\abs{\*g}_{\max}\log{(C_f/p\varepsilon)}\right)$ & & \\
     
      \hline
	\end{tabular}
	  \begin{tablenotes}[para, online,flushleft]
	\end{tablenotes}
			\end{threeparttable}}

 \label{table:resultV}
\end{table*}

\begin{table*}[ht]
	\centering
     \caption{Classical algorithms V.S. quantum algorithms of the matrix case, where $C_f$ is the curvature of the objective function $f$, $\varepsilon$ is the precision of the solution, $d$ is the dimension of the domain, $T_{\nabla}$ is the times required to evaluate $\nabla f$; $\sigma_1(M)$ and  is the largest and the second largest singular value, respectively; $r$ is the rank of the gradient matrix; ${\gamma'}_{\min}$ is a factor which depends on the relation of the singular value distribution of the gradient matrix and the direction of the initial vector.}	
	\resizebox{\linewidth}{!}{
	\centering
	\begin{threeparttable}
	\begin{tabular}{|cc|c|c|c|c|}
 \hline
		\textbf{Domain}& \textbf{Constraints}&\textbf{Algorithm}& \textbf{Iteration} & \textbf{Complexity of the Update Computing} \\
\hline
     Sparse Matrices    & $\norm{\cdot}_{tr}$-ball & FW with Power Method \cite{jaggi2013revisiting} & $O(C_f/\varepsilon)$ & $O\left(\frac{\sigma_1(M)d^2}{\varepsilon}+T_{\nabla}\right)$ \\
    &  & FW with Lanczos Method \cite{jaggi2013revisiting}  & $O(C_f/\varepsilon)$ & $O\left(\frac{\sqrt{\sigma_1(M)}d^2}{\sqrt{\varepsilon}}+T_{\nabla}\right)$\\
    &  & \textbf{FW with QTSVE (Theorem \ref{theorem:qfwqsvd})} & $O(C_f/\varepsilon)$ & $\tilde{O}\left( \frac{r \sigma_1^3(M) d}{ (\sigma_1(M)-\sigma_2(M)) \varepsilon^2}+ T_{\nabla}\right)$\\
        &  & \textbf{FW with QPM (Theorem \ref{theorem:qfwqpm})} & $O(C_f/\varepsilon)$ & $\tilde{O}\left(\frac{\sqrt{r}\sigma_1^4(M)d}{(1-\sigma_1(M)){\gamma'}_{\min}^3\varepsilon^3}+T_{\nabla}\right)$\\
      \hline
	\end{tabular}
	  \begin{tablenotes}[para, online,flushleft]
	\end{tablenotes}
			\end{threeparttable}}
 \label{table:resultM}
\end{table*}

\section{Preliminaries}
\label{sec:pre}

\subsection{Notations and Assumptions for Constrained Optimization Problem}

We consider constrained convex optimization problems of the form
\begin{align}\label{eq:opt1}
    \min_{\bx \in \cD} f(\bx),
\end{align}
where $\bx \in \R^d$, $f: \R^d \rightarrow \R$, and $\cD \subseteq \R^d$ is the constraint set. In addition, as usually in constrained convex optimization, we also make the following assumptions:
\begin{assumption}\label{asmp:func}
$f$ is convex and $L$-smooth, i.e., the gradient of $f$ satisfies $\norm{\nabla f(\bx)-\nabla f(\by)}_2 \le L\norm{\bx-\by}_2$ for any $\bx, \by \in \R^d$.

\end{assumption}

\begin{assumption}\label{asmp:constraint}
    $\cD$ is compact and convex, and the diameter of $\cD$ has an upper bound $D$, i.e., $\forall{x,y \in \mathcal{K}}, \|x-y\|_2 \leq D$. 
\end{assumption}

Typically, solving $\argmin_{\bx \in \cD} \bx^{\top} \by$ for any $\by \in \R^d$, is much faster than the projection operation onto $\cD$ (i.e., solving $\argmin_{\bx \in \cD}\norm{\bx-\by}$).
Examples of such domains include the set of sparse vectors, bounded norm matrices, flow polytope and many more \cite{hazan2012projection}.
Therefore, for such domains, the basic idea of the Frank-Wolfe algorithm is to replace the projection operation with a linear optimization problem.

In the design and analysis of the Frank-Wolfe algorithm, one important quantity is the curvature $C_f$, which measures the ``non-linearity" of $f$ and is defined as follows,
\begin{align}
    C_f =  \sup_{\bx, \bs \in \cD,\beta \in [0,1], \by=\bx+\beta(\bs-\bx) } \frac{2}{\beta^2} \times 
    \left(f(\by)-f(\bx)  -\langle \by- \bx, \nabla f(\bx)\rangle\right).
\end{align}
By Lemma 7 of \cite{jaggi2013revisiting}, the curvature can be bounded as $C_f \le L D^2$.

\subsection{Classical Frank-Wolfe Algorithm}

The classical Frank-Wolfe algorithm is given in \Cref{alg:cls_fw}.                                                               
The key step is the linear subproblem of \cref{eq:classic_linear program} which seeks an approximate minimizer in $\cD$ of $\inner{\bs, \nabla f(\bx^{(t)})}$. 
Classically, the per-step cost is $O(N)$ where $N$ is the number of elements that need to be searched which introduces a large $O(N)$ cost. In this work, we will show that only $O(\sqrt{N})$ quantum queries to solve this linear subproblem. 

\begin{algorithm}[t]
	\caption{Classical Frank-Wolfe Algorithm with Approximate Linear Subproblems}\label{alg:cls_fw}
			\resizebox{1.0\columnwidth}{!}{
\begin{minipage}{\columnwidth}
	\begin{algorithmic}[1]
	       \State {\textbf{Input:}} Solution precision $\varepsilon$, iterations $T$. 
            \State {\textbf{Output:}} $\bx^{(T)}$ such that $f(\bx^{T})-f(\bx^*)\le\varepsilon$.

	   \State \textbf{Initialize:} Let $\bx^{(1)} \in \cD$. 
	   \For{$t=1, ...,T$ }
	   \State Let $\gamma_t=\frac{2}{t+2}$.
	   \State Find direction $\bs \in \cD$ such that
	   \begin{align}
	       \langle{\bs, \nabla f(\bx^{(t)})}\rangle \le \min_{\hat{\bs} \in \cD}\langle{\hat{\bs}, \nabla f(\bx^{(t)})\rangle}+\frac{\delta}{2}\gamma_t C_f.\label{eq:classic_linear program}
	   \end{align}
	   \State Update $\bx^{(t+1)}=(1-\gamma_t) \bx^{(t)}+\gamma_t \bs$.
	    	   \EndFor
		\end{algorithmic}
		\end{minipage}}
\end{algorithm}

\begin{algorithm}[t]
	\caption{Quantum Frank-Wolfe Algorithm for Sparsity/Simplex Constraint}\label{alg:q_fw}
			\resizebox{1.0\columnwidth}{!}{
\begin{minipage}{\columnwidth}
	\begin{algorithmic}[1]
	    \State {\textbf{Input:}} Solution precision $\varepsilon$, gradient precision $\{\sigma_t\}_{t=1}^{T}$.
            \State {\textbf{Output:}} $\bx^{(T)}$ such that $f(\bx^{T})-f(\bx^*)\le\varepsilon$.
	   \State \textbf{Initialize:} Let $\bx^{(1)} \in \cD$. 
       \State Let $T=\frac{4C_f}{\varepsilon}-2$.
	   \For{$t=1, ...,T$ }
    \State Let $\gamma_t=\frac{2}{t+2}$.
    \State Prepare quantum state $\sum_{i=0}^{d-1}\ket{i}\ket{\bx^{(t)}}\ket{0}$.
    \State Perform quantum gradient circuit (Lemma \ref{Lemma:QGC}) to get $\sum_{i=0}^{d-1}\ket{i}\ket{\bx^{(t)}}\ket{\frac{f(\bx^{(t)}+\sigma_t \be_i)-f(\bx^{(t)})}{\sigma_t}}$.
    \State Apply quantum maximum finding to the absolute value of the third register (to the third register directly for the simplex constraint, respectively) (Lemma \ref{Lemma:QMF}), and then measure the first register to obtain measurement result $i_t$.
    \State Set $\bs=-\be_{i_t}$. Update $\bx^{(t+1)}=(1-\gamma_t) \bx^{(t)}+\gamma_t \bs$.
	    	   \EndFor
		\end{algorithmic}
		\end{minipage}}
\end{algorithm}

\begin{lemma}\label{lem:conv_rate}[\cite{jaggi2013revisiting}, Theorem 1] For each $t\ge 1$, the iterates of \Cref{alg:cls_fw} satisfy
\begin{align}
    f(\bx^{(t)})-f(\bx^*) \le \frac{2C_f}{t+2} (1+\delta),
\end{align}
where $x^*$ is the optimal solution to \Cref{eq:opt1}, and $\delta$ is the solution quality to which the internal linear subproblems are solved. In other words, we can use $O(\frac{(1+\delta)C_f}{\varepsilon})$ iterations to have a $\varepsilon$-opt solution.
\end{lemma}

\section{Quantum Frank-Wolfe Algorithms over Vectors}\label{sec:QFW}
\subsection{Quantum Frank-Wolfe with Sparsity Constraints}
\label{subsec:QFWSC}
We first consider the optimization problem
\begin{align}\label{eq:l1_constraint_convex}
    \min f(\bx), \text{ s.t. } \bx \in \R^d, \norm{\bx}\le 1,
\end{align}
where the sparsity constraint $\cD=\{\bx\in \R^d: \norm{\bx}_1\le 1\}$. 

For the $\ell_1$ norm problem,  the exact minimizer (i.e, corresponding to $\delta = 0$) of  \cref{eq:classic_linear program} is $\hat{\bs}=-\be_{i_t}$ with 
\begin{align}
    i_{t}&\in \argmax_{i \in [d]}|\nabla_if(\bx^{(t)})|,
\end{align}
i.e., it is a coordinate corresponding to the largest absolute value of the gradient component.

Our approach will be to construct an approximate quantum maximum gradient component finding algorithm to find such an $i_t$.

\textbf{Quantum access model $\bU_f$.}
In this subsection, 
we assume that the value of the loss function is accessed via a function value oracle as shown in Assumption \ref{asmp:oracle}.

\begin{assumption}\label{asmp:oracle}
There is a unitary $\bU_f$ that, in time $T_f$, returns the function value, i.e.,  $\bU_f:\ket{\bx}\ket{a}\rightarrow \ket{\bx}\ket{a + f(\bx)}$, for any $a$, where $\ket{\bx}\defeq \ket{x_1}\ket{x_2}...\ket{x_d}$.
\end{assumption}

\textbf{Quantum gradient circuit.}
Next, we present a general unitary $U_g$ to approximate the gradient $\nabla f(\bx_t)$. Specifically, we use the forward difference $g_i(\bx_t)=\frac{f(\bx_t+\sigma \be_i)-f(\bx_t)}{\sigma}$ to approximate each item of $\nabla_i f(\bx_t)$ with $\ell_{\infty}$ error $\varepsilon_{g}$, i.e., $\norm{\nabla f(\bx_t)-g(\bx_t)}_{\infty} \le \varepsilon_{g}$, where $\sigma$ is the tunable parameter for the desired accuracy.

\begin{lemma}[Theorem 3.1 \cite{berahas2022theoretical}]
\label{Lemma:FDM}
Under \Cref{asmp:func}, let $g_i(\bx)=\frac{f(\bx+\sigma \be_i)-f(\bx)}{\sigma}$, then for all $\bx \in \R^d$,
\begin{align}
    \norm{g(\bx)-\nabla f(\bx)}_2 \le \frac{\sqrt{d}L\sigma}{2}.
\end{align}
\end{lemma}

\begin{restatable}{lemma}{LemmaQGC}
\label{Lemma:QGC}
    Given access to the quantum function value oracle $\bU_f$, there exists a quantum circuit to construct a quantum error bounded gradient oracle $\bU_g: \ket{i}\ket{\bx}\ket{0}\rightarrow \ket{i}\ket{\bx}\ket{g_i(\bx)}$, where $g_i(\bx)=\frac{f(\bx+\sigma \be_i)-f(\bx)}{\sigma}$ is the $i$-th component of the gradient and $\sigma$ is the tunable parameter, with two queries to the quantum function value oracle.
\end{restatable}

The proof is given in Appendix \ref{proof:Lemma:QGC}.

\textbf{Quantum maximum finding circuit.} 
Based on $\bU_g$, leveraging the generalized quantum min-finding algorithm with an approximate unitary (Theorem 2.4 in \cite{chen2021quantum}), we can obtain an approximate search $A_{\max}$ of the maximum gradient component as shown in Lemma \ref{Lemma:QMF}, with a simplified proof given in Appendix \ref{proof:Lemma:QMF}.
\begin{restatable}{lemma}{LemmaQMF}(Approximate maximum gradient component finding) 
    \label{Lemma:QMF}
    Given access to the quantum error bounded gradient oracle $\bU_g: \ket{i}\ket{\bx}\ket{0}\rightarrow \ket{i}\ket{\bx}\ket{g_i(\bx)}$ s.t. for each $i \in [d]$, after measuring $\ket{g_i(\bx)}$, the measured outcome $g_i(\bx)$ satisfies $|g_i(\bx)-\nabla f_i(\bx)|\le \epsilon$.
    There exists a quantum circuit $\cA_{\max}$ that finds the index $i^*$ that satisfies $\nabla f_{i^*}(\bx)\ge \max_{j \in [d]}\nabla f_j(\bx)-2\epsilon$ or $\abs{\nabla f_{i^*}(\bx)}\ge \max_{j \in [d]}\abs{\nabla f_j(\bx)}-2\epsilon$, using $O(\sqrt{d}\log(\frac{1}{\delta}))$ applications of $\bU_g$, $\bU_g^{\dagger}$ and $O(\sqrt{d})$ elementary gates, with probability $1-\delta$.
\end{restatable}

\textbf{Convergence Analysis.}
Now we can conduct the convergence analysis with the help of approximate maximum finding sub-routine and show how to choose appropriate parameters, which gives Theorem \ref{Theorem:QFWSC}, with proof given in Appendix \ref{proof:Theorem:QFWSC}.

\begin{restatable}{theorem}{TheoremQFWSC}(Quantum FW over the sparsity constraint) 
\label{Theorem:QFWSC}
By setting $\sigma_t=\frac{C_f}{\sqrt{d}L(t+2)}$ for $t\in[T]$,  the quantum algorithm (Algorithm \ref{alg:q_fw}) solves the sparsity constraint optimization problem for any precision $\varepsilon$ such that $f(\bx^{T})-f(\bx^*)\le\varepsilon$ in $T=\frac{4C_f}{\varepsilon}-2$ rounds, succeed with probability $1-p$, with $O\left(\sqrt{d} \log{\frac{C_f}{p\varepsilon}} \right)$ calls to the function value oracle $\bU_f$ per round.
\end{restatable}

If the objective function is a $G$-Lipschitz continues function (i.e. $|f(\bx)-f(\by)| \leqslant G \|\by-\bx\|,\quad \forall \bx, \by \in \cD$), an alternative approach for estimating the gradient of the objective function involves employing the bounded-error Jordan algorithm to improve the query complexity of each iteration to $O(1)$, at the cost of additional space complexity and extra gate operations.
This result is given in Appendix \ref{sec:QFWJ}.

\subsection{Extensions: Quantum Frank-Wolfe for Atomic Sets}

Classically, the Frank-Wolfe algorithm has been shown to be well-suited to atomic sets~\cite{jaggi2013revisiting}, i.e.\ where the constraint set is expressed as the convex hull of another (not-necessarily finite) set $\+A$: 
$ \+D = \mathsf{conv}(\+A) $
In this case, the Frank-Wolfe update calculation requires a minimization only over $\+A$: 
$\min_{\hat{\bs} \in \+A}\langle \hat{\bs}, \nabla f(x^{(t)})\rangle. $

The optimization over the $\ell_1$ ball as studied above is a special case of this, since 
\begin{align}
\{ \bx\in\mb{R}^d : \norm{\bx}_1\} = \mathsf{conv}\{\pm\be_1, \pm\be_2, \ldots, \pm\be_d\}.
\end{align}
Note also that quantum optimization over the simplex $\Delta_d=\mathsf{conv}\{\be_1,\ldots, \be_d\}$ can be done by almost exactly the same method as for the $\ell_1$ case, with the only modification to account for the fact that only the unit vectors need to be optimized over, which gives Theorem \ref{Theorem:QFWS}.

\begin{theorem}(Quantum FW over the simplex) 
\label{Theorem:QFWS}
By setting $\sigma_t=\frac{C_f}{\sqrt{d}L(t+2)}$ for $t\in[T]$, the quantum algorithm (Algorithm \ref{alg:q_fw}) solves the simplex constraint optimization problem for any precision $\varepsilon$ such that $f(\bx^{T})-f(\bx^*)\le\varepsilon$ in $T=\frac{4C_f}{\varepsilon}-2$ rounds, succeed with probability $1-p$, with $O\left(\sqrt{d} \log{\frac{C_f}{p\varepsilon}} \right)$ calls to the function value oracle $\bU_f$ per round.
\end{theorem}

Two more extensions for atomic sets will be given in Appendix \ref{sec:extension}.

\section{Quantum Frank-Wolfe Algorithms over Matrices}
\label{Sec:QFWM}

In this section, we consider the matrix version of the constrained optimization problem in \cref{eq:general_constraint_convex}, specifically,
\begin{align}\label{eq:l2_constraint_convex}
    \min f(X), \text{ s.t. } X \in \R^{d\times d}, \norm{X}_{\tr}\le 1,
\end{align}
where the sparsity constraint is $\cD=\{X\in \R^{d\times d}: \norm{X}_{\tr}\le 1\}$. 
For simplicity of presentation, we first focus on square matrices, i.e., $X\in \R^{d\times d}$ (Remark \ref{remark:square}).

\textbf{Schatten matrix norm.} In contrast to the vector norm $\norm{\cdot}$ on $\R^d$, the corresponding Schatten matrix norm $\norm{X}$ is defined as $\norm{(\sigma_1, ..., \sigma_d)}$, where $\sigma_1, ..., \sigma_d$ are singular values of $X$. 
It is known that the dual of the Schatten $\ell_p$ norm is the Schatten $\ell_q$ norm with $1/p+1/q=1$. 
The most prominent example is the trace norm $\norm{\cdot}_{\tr}$, also referred to as the nuclear norm or Schatten $\ell_1$ norm, defined as the sum of the singular values $\norm{X}_{\tr}=\sum_{i=1}^{d} \sigma_i$.

\textbf{Linear subproblem solver.} Following the classical Frank-Wolfe iteration framework, we aim to solve the linear optimization subproblem $\min_{S \in \cD}\inner{S,\nabla f(X_t)}$ where $X_t$ denotes the iterate matrix at step $t$, and $\langle X, Y\rangle =\tr X^\top Y$ represents the Hilbert-Schmidt inner product. 
For convenience, let $M=\nabla f(X_t)$ in the rest of this section.  
To solve this subproblem, one can compute the singular value decomposition (SVD) $M=U\text{diag}(\boldsymbol{\sigma})V^{\top}$, where $\boldsymbol{\sigma}$ are singular values of $M$ and $U,V\in \R^{d\times d}$ are orthogonal matrices. 
Since Schatten norms are invariant under orthogonal transformations, the optimal solution $S \in \cD$ for the minimization problem $\min_{S \in \cD}\inner{S,M}$ takes the forms of $S=U\text{diag}(\bs)V^{\top}$, where $\inner{\bs,\boldsymbol{\sigma}}=\norm{\boldsymbol{\sigma}}_q$ with $\norm{\bs}_p\le 1$ and $1/p+1/q=1$. 
For the nuclear norm (i.e., $\ell_1$ Schatten norm), this reduces to $S=\bu \bv^{\top}$ where $\bu,\bv$ are the left and right top singular vectors of $M$, corresponding to its largest singular value $\sigma_1(M)$.
Thus, the core computational task is to efficiently approximate the top singular vectors $\bu,\bv \in R^d$, ensuring $|\bu^\top M \bv - \sigma_1(M)| \le \varepsilon $.

\textbf{Power method and Lanczos method.} Compared with the SVD that requires $O(d^3)$ computational cost per iteration to compute all $d$ singular vectors, extracting only the top singular vector is much easier. 
Specifically, \cite{kuczynski1992estimating} considers two iterative methods: the power method and the Lanczos method. 
The power method achieves $|\bu^\top M \bv - \sigma_1(M)| \le \varepsilon'$ with the worst-case computation complexity of $O\left(\frac{\sigma_1(M) d^2 \ln d}{\varepsilon'}\right)$, while the Lanczos method achieves $|\bu^\top M \bv - \sigma_1(M)| \le \varepsilon'$ with the worst-case computation complexity of $O\left(\frac{\sqrt{\sigma_1(M)}  d^2 \ln d}{\sqrt{\varepsilon'}}\right)$, where $\varepsilon'$ is the additive error.
Similar to the convergence analysis in \cref{subsec:QFWSC}, setting $\varepsilon'=O(\varepsilon)$, the complexity of update computing are $O\left(\frac{\sigma_1(M) d^2 \ln d}{\varepsilon}\right)$ and $O\left(\frac{\sqrt{\sigma_1(M)}  d^2 \ln d}{\sqrt{\varepsilon}}\right)$, respectively. 

\textbf{Quantum enhancement.} In the following, we propose two quantum subroutines to compute the top singular vector: the quantum top singular vector extraction method and the quantum power method.
First, we assume the following gradient access model for matrix data.
A detailed description of this data structure can be found in Section 1.A of \cite{kerenidis2020quantum}.

\begin{assumption}[Quantum access to a matrix]\label{asmp:QMA}
We assume that we have efficient quantum access to the matrix $M\in R^{d \times d}$. 
That is, there exists a data structure that allows performing the mapping 
$\ket{i}\ket{0}\to \ket{i}\ket{M_{i,\cdot}}=\ket{i} \frac{1}{\norm{M_{i,\cdot}}}\sum_j^d M_{ij}\ket{j}$ 
for all $i$, 
and $\ket{0}\to \frac{1}{\norm{M}_F} \sum_i^d \norm{M_{i,\cdot}}\ket{i}$ 
in time $\widetilde{O}(1)$.
\end{assumption}

\subsection{Quantum Frank-Wolfe with Quantum Top Singular Vector Extraction}
\label{subsec:qfwQSVE}

Leveraging the quantum access defined in Assumption \ref{asmp:QMA}, quantum singular value estimation can be performed efficiently.

\begin{lemma}[Singular value estimation (Theorem 3 \cite{bellante2022quantum}), \cite{kerenidis2020quantum}]
    \label{Lemma:QSVE}
    Let there be quantum access to $M\in R^{d \times d}$, with singular value decomposition $M=\sum_{i}^{d}\sigma_i \bu_i \bv_i^T$. Let $\epsilon>0$ be a precision parameter. There exists a quantum circuit for performing the mapping $\frac{1}{\norm{M}_F}\sum_i^d\sum_j^d M_{ij} \ket{i}\ket{j}\ket{0} \to \frac{1}{\norm{M}_F}\sum_i^k \sigma_i \ket{\bu_i} \ket{\bv_i}\ket{\overline{\sigma}_i}$ such that $\abs{\sigma_i-\overline{\sigma}_i}\leq \epsilon$ with probability at least $1-1/\text{poly}(d)$ in time $O\left(\frac{\norm{M}_F \text{poly}\log{d}}{\epsilon}\right)$.
\end{lemma}

To extract classical singular vectors corresponding to the largest singular value from a quantum state, $\ell_2$ norm quantum state tomography is required.

\begin{lemma}[$\ell_2$ state-vector tomography \cite{kerenidis2020quantum2,kerenidis2019quantum}]
    \label{Lemma:L2T}
    Given a unitary mapping $U_x:\ket{0}\to \ket{\bx}$ in time $T(U_{\bx})$ and $\delta>0$, there is an algorithm that produces an estimate $\overline{\bx}\in R^d$ with $\norm{\overline{\bx}}_2=1$ such that $\norm{\bx-\overline{\bx}}_2\leq \delta$ with probability at least $1-1/\text{poly}(d)$ in time $O\left(T(U_{\bx}) \frac{d \log{d}}{\delta^2}\right)$.
\end{lemma}

\textbf{Quantum top singular vector extraction (QTSVE).} The goal of the quantum subroutine in each iteration is to find the top right / left singular vectors of the gradient matrix. 
First, we prepare the gradient matrix state using the quantum access as stated in Assumption \ref{asmp:QMA}, then we perform QSVE to this state. 
The quantum maximum finding is applied to obtain the quantum state corresponding to the largest singular value. 
Prepare sufficient quantum states corresponding to the largest singular value until satisfying the requirement of tomography, then perform quantum state tomography to extract the corresponding right / left classical singular vectors. 
This procedure is shown in Lemma \ref{Lemma:QGSVE}, with the proof given in Appendix \ref{proof:Lemma:QGSVE}.
Note that the success probability of QTSVE can be improved by repeating it logarithmic times and then taking the average. 

\begin{restatable}{lemma}{LemmaQGSVE}(Quantum top singular vector extraction)
    \label{Lemma:QGSVE}
    Let there be efficient quantum access to a matrix $M\in R^{d \times d}$, with singular value decomposition $M=\sum_i^d \sigma_i \bu_i \bv_i^T$. 
    Define $p=\frac{\sigma^2_1(M)}{\sum_{i=1}^d \sigma_i^2}$.
    There exist quantum algorithms that with time complexity $O\left( \frac{\norm{M}_F d  \text{poly}\log{d}}{\sqrt{p} \epsilon \delta^2} \right)$, give the estimated top singular value $\overline{\sigma}_1$ of $M$ to precision $\epsilon$ and the corresponding unit estimated singular vectors $\bu, \bv$ to precision $\delta$ such that $\norm{\bu-\bu_{top}}\leq \delta$, $\norm{\bv-\bv_{top}}\leq \delta$ with probability at least $1-1/\text{poly}(d)$.
\end{restatable}

\textbf{Convergence Analysis.} 
Our quantum Frank-Wolfe algorithm for nuclear norm constraint (Algorithm \ref{alg:qfwm1}) then follows, with the analysis given in Appendix \ref{proof:Theorem:QFWQSVD}.

\begin{restatable}{theorem}{theoremqfwqsvd}(Quantum FW with QTSVE)\label{theorem:qfwqsvd}
    By setting $\delta_t=\frac{C_f}{2(t+2)\sigma_1(M_t)}$ and $\epsilon_t\leq (\sigma_1(M_t)-\sigma_2(M_t))/2$ for $t\in[T]$, the quantum algorithm (Algorithm \ref{alg:qfwm1}) solves the nuclear norm constraint optimization problem for any precision $\varepsilon$ such that $f(X^{T})-f(X^*)\le\varepsilon$ in $T=\frac{4C_f}{\varepsilon}-2$ rounds, with time complexity $\tilde{O}\left( \frac{r \sigma_1^3(M) d}{(\sigma_1(M)-\sigma_2(M)) \varepsilon^2} \right)$ for computing the update direction per round, where $r$ is the rank of the gradient matrix.
\end{restatable}

\begin{algorithm}[t]
	\caption{Quantum Frank-Wolfe Algorithm for Nuclear Norm Constraint with QTSVE}\label{alg:qfwm1}
			\resizebox{1.0\columnwidth}{!}{
\begin{minipage}{\columnwidth}
	\begin{algorithmic}[1]
            \State {\textbf{Input:}} Solution precision $\varepsilon$, singular value precision $\{\epsilon_t\}_{t=1}^{T}$, tomography precision $\{\delta_t\}_{t=1}^{T}$.
            \State {\textbf{Output:}} $X^{(T)}$ such that $f(X^{T})-f(X^*)\le\varepsilon$.
	    \State \textbf{Initialize:} Let $X^{(1)} \in \cD$. 
        \State Let $T=\frac{4C_f}{\varepsilon}-2$.
            \For{$t=1, ...,T$ }
                \State Let $\gamma_t=\frac{2}{t+2}$.
                \State Prepare $\frac{1}{\norm{M}_F}\sum_i^d\sum_j^d M_{ij} \ket{i}\ket{j}\ket{0}$.
                \State Perform QSVE (Lemma \ref{Lemma:QSVE}) to get $\frac{1}{\sqrt{\sum_i^r}\sigma_i^2} \sum_i^r \sigma_i \ket{\bu_i}\ket{\bv_i} \ket{\overline{\sigma}_i}, where |\sigma_i-\overline{\sigma}_i|\leq \epsilon_t$.
                \State Apply quantum maximum finding (Lemma \ref{Lemma:QMF}) to the third register to get $\ket{\bu_{top}}\ket{\bv_{top}}\ket{\overline{\sigma}_1}$.
                \State Perform $\ell_2$-norm tomography (Lemma \ref{Lemma:L2T}), to obtain $\bu$, $\bv$, where $\norm{\bu-\bu_{top}}\leq\delta_t$, $\norm{\bv-\bv_{top}}\leq\delta_t$.
                \State Set $S=\bu\bv^{\top}$. Update $X^{(t+1)}=(1-\gamma_t) X^{(t)}+\gamma_t S$.
            \EndFor
	\end{algorithmic}
	\end{minipage}}
\end{algorithm}

In computing the update direction, Algorithm \ref{alg:qfwm1} reduces a $O(d\varepsilon(\sigma_1(M)-\sigma_2(M))/r \sigma^2_1(M))$ factor to the power method and $O(d{\varepsilon}^{1.5}(\sigma_1(M)-\sigma_2(M))/r\sigma^{2.5}_1(M))$ to the Lanczos method, respectively.

\subsection{Quantum Frank-Wolfe with Quantum Power Method}
\label{subsec:qfwQPM}
The second framework is to accelerate the power method directly with quantum matrix-vector multiplication method and quantum tomography. 
The classical power method constructs a sequence $\bz_0, ..., \bz_k$, where $\bz_0=\boldsymbol{b}$ is drawn uniformly random over a unit sphere ${\boldsymbol{b}:\norm{\boldsymbol{b}}_2=1}$, and $\bz_{i+1}=M^{\top}M\bz_{i}$ for $i\ge 1$, ($\bz_{i+1}=MM^{\top}\bz_{i}$ for the left singular vector, respectively).
After $k=\frac{C_0\sigma_1(M)\ln d}{\varepsilon}$, we have $\abs{\frac{\bz_k^{\top}M\bz_k}{\norm{\bz_k}_2^2}-\sigma_1(M)}\le \varepsilon$, 
where $C_0$ is a constant. 

\textbf{Quantum power method (QPM).} Using the quantum access given in Assumption \ref{asmp:QMA}, the quantum matrix-vector multiplication can be performed efficiently:

\begin{lemma}(Quantum matrix-vector multiplication (Theorem 4 \cite{bellante2022quantum}), \cite{chakraborty2018power})\label{lem:q_multiply}
    Let there be quantum access to the matrix $M\in R^{d \times d}$ with $\sigma_{max}\leq 1$, and to a vector $\bz\in R^d$. Let $\norm{M\bz}\geq \gamma'$. There exists a quantum algorithm that creates a state $\ket{\by}$ such that $\norm{\ket{\by}-\ket{M\bz}}\leq \epsilon$ in time $\tilde{O}\left(\frac{1}{\gamma'} \norm{M}_F \log(1/\epsilon)\right)$, with probability at least $1-1/poly(d)$.
\end{lemma}

Apply $2k$ times of quantum matrix-vector multiplication, we can get a quantum state corresponding to $\bz_k$, as shown in Lemma \ref{lemma:qpm}, with proof given in Appendix \ref{proof:lemma:qpm}.
A similar process can be constructed to compute $(MM^{\top})^k\bb$ (corresponding to the left singular vector) simultaneously. 

\begin{restatable}{lemma}{lemmaqpm}(Quantum power method)\label{lemma:qpm}
    Let there be quantum access to the matrix $M\in R^{d \times d}$ with $\sigma_{max}\leq 1$, and to a vector $\bz\in R^d$. Let ${\gamma'}_{\min}$ be the lower bound of $\norm{(M^{\top}M)^i \bz)}$ for all $i\in [k]$. There exists a quantum algorithm that creates a state $\ket{\by}$ such that $\norm{\ket{\by}-\ket{(M^{\top}M)^k\bz}}\leq \delta$ in time $\tilde{O}(\frac{k}{{\gamma'}_{\min}} \norm{M}_F \log(1/\delta))$, with probability at least $1-O(k/poly(d))$.
\end{restatable}

\begin{algorithm}[t]
	\caption{Quantum Frank-Wolfe Algorithm for Nuclear Norm Constraint with QPM}\label{alg:qfwm2}
			\resizebox{1.0\columnwidth}{!}{
\begin{minipage}{\columnwidth}
	\begin{algorithmic}[1]
            \State {\textbf{Input:}} Solution precision $\varepsilon$, multiplication times $\{k_t\}_{t=1}^T$, multiplication precision $\{\delta_t\}_{t=1}^T$, tomography precision $\{\delta'_t\}_{t=1}^T$.
            \State {\textbf{Output:}} $X^{(T)}$ such that $f(X^{T})-f(X^*)\le\varepsilon$.
	    \State \textbf{Initialize:} Let $X^{(1)} \in \cD$. 
            \State Let $T=\frac{4C_f}{\varepsilon}-2$.
            \For{$t=1, ...,T$ }
                \State Let $\gamma_t=\frac{2}{t+2}$.
                \State Prepare $\frac{1}{\norm{M}_F}\sum_i^d\sum_j^d M_{ij} \ket{i}\ket{j}\ket{\bb}\ket{\bb}$, where $\bb$ is the uniform superposition state.
                \State Apply quantum power method (Lemma \ref{lemma:qpm}) to get $\frac{1}{\norm{M}_F}\sum_i^d\sum_j^d M_{ij} \ket{i}\ket{j}\ket{\overline{\bz}_u}\ket{\overline{\bz}_v}$, where $\norm{\overline{\bz}_u-(MM^{\top})^k \bb}\le \delta_t, \norm{\overline{\bz}_v-(M^{\top}M)^k \bb} \le \delta_t$.
                \State Perform $\ell_2$-norm tomography (Lemma \ref{Lemma:L2T}) to obtain $\bu,\bv$, where $\norm{\bu-\overline{\bz}_u}\le\delta'_t,\norm{\bv-\overline{\bz}_v}\le\delta'_t$.
                \State Set $S=\bu\bv^{\top}$. Update $X^{(t+1)}=(1-\gamma_t) X^{(t)}+\gamma_t S$.
            \EndFor
	\end{algorithmic}
	\end{minipage}}
\end{algorithm}

\textbf{Convergence Analysis.} After quantum state tomography, we can extract the classical top singular vectors.
Note that the success probability of QPM and tomography can be improved by repeating the whole procedure logarithmic times and then taking the average.
Our quantum Frank-Wolfe algorithm (Algorithm \ref{alg:qfwm2}) for nuclear norm constraint then follows, and the parameters choosing and convergence analysis are given in Theorem \ref{theorem:qfwqpm}, with the proof given in Appendix \ref{proof:theoremqfwqpm}.

\begin{restatable}{theorem}{theoremqfwqpm}(Quantum FW with QPM)\label{theorem:qfwqpm}
    By setting $k_t=\frac{2C_0\sigma_1(M_t)\ln d}{\varepsilon},\delta_t=\delta'_t=\frac{\varepsilon {\gamma'}_{\min}}{16 \sigma_1(M)}$ for $t\in[T]$, the quantum algorithm (Algorithm \ref{alg:qfwm2}) solves the nuclear norm constraint optimization problem for any precision $\varepsilon$ such that $f(X^{T})-f(X^*)\le\varepsilon$ in $T=\frac{4C_f}{\varepsilon}-2$ rounds, with time complexity $\tilde{O}\left(\frac{\sqrt{r}\sigma_1^4(M)d}{(1-\sigma_1(M)){\gamma'}_{\min}^3\varepsilon^3}\right)$ for computing the update direction per round, where $r$ is the rank of the gradient matrix, $C_0$ is a constant and ${\gamma'}_{\min}$ is the lower bound of $\norm{(M^{\top}M)^i \bb)}$ for all $i\in [k]$.
\end{restatable}

In computing the update direction, Algorithm \ref{alg:qfwm2} reduces a $O(d\varepsilon^2{\gamma'}^3_{\min}/\sqrt{r} \sigma^3_1(M))$ factor to the power method and $O(d{\varepsilon}^{2.5}{\gamma'}^3_{\min}/\sqrt{r}\sigma^{3.5}_1(M))$ to the Lanczos method. 
A discussion of this section is given in Appendix \ref{sec:discussion}.

\section{Conclusion and Future Work}
\label{sec:conclude}
This paper addresses the projection-free sparse convex optimization problem. 
We propose several quantum Frank-Wolfe algorithms for both vector and matrix domains, demonstrating the quantum speedup over the classical methods with respect to the dimension of the feasible set.

For future work, we aim to extend quantum Frank-Wolfe methods to stochastic and online optimization frameworks, with a focus on characterizing quantum advantages in projection-free regret minimization.
Meanwhile, \cite{jaggi2013revisiting} highlights several interesting cases involving matrix norms, where classical approaches often rely on computationally expensive singular value decomposition. 
A potential avenue of interest is determining whether quantum singular value estimation can yield greater speedups in such settings.
These investigations would collectively advance the understanding of quantum-enhanced projection-free optimization in high-dimensional spaces.

\section*{Acknowledgement}
The authors wish to thank Xutong Liu and Jonathan Allcock for valuable discussion.

\bibliographystyle{plain}
\bibliography{main}

\newpage
\appendix
\onecolumn

\textbf{Appendix}

\section{Extension and Discussion}

\subsection{Quantum Frank-Wolfe over Vectors with Bounded-error Jordan Algorithm}
\label{sec:QFWJ}

The quantum Frank-Wolfe Algorithm with Bounded-error Jordan’s Algorithm is shown in Algorithm \ref{alg:q_fwJ}.
We reformulate the results of the bounded-error Jordan algorithm from \cite{he2024quantum} in terms of infinity norm error, with the proof detailed in the Appendix \ref{proof:QGE}.

\begin{algorithm}[ht]
	\caption{Quantum Frank-Wolfe Algorithm with Bounded-error Jordan Algorithm}\label{alg:q_fwJ}
			\resizebox{1.0\columnwidth}{!}{
\begin{minipage}{\columnwidth}
	\begin{algorithmic}[1]
	    \State {\textbf{Input:}} Solution precision $\varepsilon$, gradient precision $\{\sigma_t\}_{t=1}^{T}$.
            \State {\textbf{Output:}} $\bx^{(T)}$ such that $f(\bx^{T})-f(\bx^*)\le\varepsilon$.
	   \State \textbf{Initialize:} Let $\bx^{(1)} \in \cD$. 
       \State Let $T=\frac{4C_f}{\varepsilon}-2$.
	   \For{$t=1, ...,T$ }
    \State Let $\gamma_t=\frac{2}{t+2}$.
    \State Using Algorithm \ref{Alg:JQGE} to get the whole vector of estimated gradient $\widetilde{\nabla} f_t(\bx_t)$.
    \State Scan all the component of $\widetilde{\nabla} f_t(\bx_t)$ to find the coordinate $i_t$ corresponding to the largest absolute value of the estimated gradient component.
    \State Set $\bs=-\be_{i_t}$. Update $\bx^{(t+1)}=(1-\gamma_t) \bx^{(t)}+\gamma_t \bs$.
	    	   \EndFor
		\end{algorithmic}
		\end{minipage}}
\end{algorithm}

\begin{restatable}{lemma}{QGB}(Lemma 1 \cite{he2024quantum})
\label{QGB}
    If $f$ is $G$-Lipschitz  continues and $L$-smooth convex function and can be accessed by a quantum function value oracle, then there exists an quantum algorithm that for any $r>0$ and $1\geq \rho >0$, gives the estimated gradient $g(x)$, which satisfies
    \begin{align}
         \Pr[\|{g(x)-\nabla f(x)}\|_{\infty} >8 \pi n^2 (n/\rho +1) L r/\rho] <  \rho, 
    \end{align}
    using $O(1)$ applications of $\bU_f$ and $O(d \log d)$ elementary gates. The space complexity is $O\left(d \log{\frac{G\rho}{4 \pi d^2 L r}}\right)$.
\end{restatable}

The next step is to determine the quantum gradient estimated parameters $r_t$ in each Frank-Wolfe iteration through convergence analysis.

\begin{restatable}{theorem}{TheoremQFWJ}(Quantum FW with bounded-error Jordan algorithm) 
\label{Theorem:QFWJ}
By setting $r_t=\frac{\rho  C_f}{16 \pi d^2 (d/\rho +1) L (t+2)}$ for $t\in[T]$,  the quantum algorithm (Algorithm \ref{alg:q_fwJ}) solves the sparsity constraint optimization problem for any precision $\varepsilon$ such that $f(\bx^{T})-f(\bx^*)\le\varepsilon$ in $T=\frac{4C_f}{\varepsilon}-2$ rounds, with $O\left(1\right)$ calls to the function value oracle $\bU_f$ per round.
\end{restatable}
The proof is given in Appendix \ref{proof:Theorem:QFWJ}. Substituting the parameter $r_t$ into the space complexity yields the qubit requirement as $O\left(d \log{\frac{Gd}{\rho\varepsilon}}\right)$.
Since each gradient estimation succeeds with probability $1-\rho$, the probability that all $T$ iterations succeed is at least $1-T\rho$. By setting $\rho=p/T$, we ensure an overall success probability of at least $1-p$.

\subsection{More Extensions over Vectors for Atomic Sets}
\label{sec:extension}
In this appendix, we give two more extension for the vector case. The first extension is to consider $\abs{\+A} = N$, with each $a_j\in \+A$ being $\tau$-sparse with non-zero $(index, value)$ pairs $(i_k,(a_j)_k)$, i.e., each $a_j\in\mb{R}^d$, but has only $\tau$ non-zero elements.
Assume that the non-zero elements are accessed with a quantum oracle $V$ which implements the transformation $V\ket{j}\ket{k}\ket{0}\ket{0}\rightarrow \ket{j}\ket{k}\ket{i_k}\ket{(a_j)_k}$.
One can construct a coherent access to the non-zero elements
\eq{ 
V^{\otimes \tau}\ket{j}\bigotimes_{k=1}^\tau\ket{k}\ket{0}\ket{0} &=
\ket{j}\bigotimes_{k=1}^\tau\ket{k}\ket{i_k}\ket{(a_j)_k} 
} 
using $\tau$ calls of $V$.
Then, a slight modification of the method of Section~\ref{subsec:QFWSC} can compute the FW update using $O(\tau\sqrt{N}\log(1/\delta))$ queries to $V$ and $U_g$.

The second extension is to consider latent group norm constraints, which have found use in inducing sparsity in problems in machine learning~\cite{jenatton2011structured}. 
The $\ell_1$ norm, $d$-simplex, group lasso etc. are all special cases of this.

Following~\cite{jaggi2013revisiting} we let $\+G = \{\*g_1, \*g_2, \ldots,\*g_{\abs{\+G}}\}$, $\*g_i\subseteq[d]$, $\bigcup_{i}\*g_i = [d]$. 
Note that the $\*g_i$ need not be disjoint. 
For each $\*g\in\+G$, let $\norm{\cdot}_{\*g}$ be an arbitrary $\ell_p$ norm, and define the \textbf{latent group norm} 
\eq{
\norm{x}_{\+G}&:= \min_{v_{(\*g)}\in\mb{R}^{\abs{\*g}}}\sum_{\*g\in\+G}\norm{v_{(\*g)}}_{\*g} \\
\text{s.t.}\quad x&=\sum_{\*g\in\+G}v_{[\*g]}
}
where $v_{(\*g)}\in\mb{R}^{\*g}$ is the restriction of $v\in\mb{R}^d$ to coordinates in $\*g$, and $v_{[\*g]}\in\mb{R}^d$ has zeros outside the support of $\*g$.  In this case, the Frank-Wolfe update corresponds to finding the value $s : \norm{s}_\+G \le 1$ such that $s^\top \nabla f(x) = \norm{\nabla f(x)}_\+G^*$, where $\norm{\nabla f(x)}_{\+G}^*=\max_{\bs:\norm{\bs}_\+G \le 1} s^\top \nabla f(\bx)$. 

By Section 4.1 in \cite{jaggi2013revisiting}, this norm is an atomic norm, and the dual norm is given by
\begin{equation}
\norm{\nabla f(x)}_{\+G}^* = \max_{\*g\in\+G}\norm{\nabla f(x)_{(\*g)}}_{\*g}^*,
\end{equation}
which implies that
\begin{align}
    \max_{s : \norm{s}_\+G \le 1}(-s^\top \nabla f(x)) = \max_{\*g\in \+G} \max_{s : \norm{s}_{\*g}\le 1}(-s^\top \nabla f(x)).
\end{align}
Therefore, it suffices to consider each $\norm{\cdot}_{\*g}$ ball separately, and then do quantum maximizing over all the $\abs{\+G}$ balls to find the one that has the largest value of $\norm{-\nabla f(x)_{\*g_i}}_{p_i}^*$.
The quantum Frank-Wolfe algorithm over latent group norm ball is then given in Algorithm \ref{alg:q_fwlg}. 
Note that by the absolute homogeneity property of dual norms, 
\begin{align}
    \norm{\nabla f(x)_{(\*g)}}_{\*g}^*=\norm{-\nabla f(x)_{(\*g)}}_{\*g}^*,
\end{align}
certain negative signs have been omitted in the algorithmic formulation.

\begin{algorithm}[ht]
	\caption{Quantum Frank-Wolfe Algorithm over Latent Group Norm Ball}\label{alg:q_fwlg}
			\resizebox{1.0\columnwidth}{!}{
\begin{minipage}{\columnwidth}
    \begin{algorithmic}[1]
	    \State {\textbf{Input:}} Gap $\varepsilon$, accuracy $\{\sigma_t\}_{t=1}^{T}$, iterations $T$.
	   \State \textbf{Initialize:} Let $\bx^{(1)} \in \cD$. 
    \For{$t=1, ...,T$ }
        \State Let $\gamma_t=\frac{2}{t+2}$, $\bx=\bx^{(t)}$.
        \State Prepare state $\sum_{i=1}^n\ket{i}_A\ket{\bx}\bigotimes_{j=1}^{\abs{\*g_i}}\ket{\*g_{i,j}} \ket{0}\ket{0}\ket{0}\ket{0}$.
        \State Perform quantum gradient circuit to get $\sum_{i=1}^n\ket{i}_A\ket{\bx}\bigotimes_{j=1}^{\abs{\*g_i}}\ket{\*g_{i,j}}\ket{g_{\*g_{i,j}}(\bx)} \ket{0}\ket{0}\ket{0}$, where $g_{\*g_{i,j}(\bx)}=\frac{f(\bx+\sigma_t \be_{\*g_{i,j}})-f(\bx)}{\sigma_t}$
        \State Compute $\sum_{i=1}^n\ket{i}_A\ket{\bx} \lp \bigotimes_{j=1}^{\abs{\*g_i}}\ket{\*g_{i,j}}\ket{g_{\*g_{i,j}}(\bx)}\ket{\mathsf{sgn}(g_{\*g_{i,j}}(\bx))\abs{g_{\*g_{i,j}}(\bx)}^{q_i-1}}\rp $ $ \ket{\norm{g(\bx)_{(\*g_i)}}_{p_i}} \ket{\norm{g(\bx)_{(\*g_i)}}_{p_i}^*}$.
        \State Apply quantum maximum finding on the last register, and then measure the rest registers, denote the result of the first register as $i_t$.
        \State Initial $\bs=0$, set $\bs_{\*g_{i_t,j}}=\mathsf{sgn}(g_{\*g_{i_t,j}}(\bx))\abs{g_{\*g_{i_t,j}}(\bx)}^{q_{i_t}-1}$ for $j=1$ to $\abs{\*g_i}$, where $\frac{1}{p_{i_t}}+\frac{1}{q_{i_t}}=1$. Then normalize $\bs$.
        \State Update $\bx^{(t+1)}=(1-\gamma_t) \bx^{(t)}+\gamma_t \bs$.
    \EndFor
    \end{algorithmic}
\end{minipage}}
\end{algorithm}

To simplify the proof of the query complexity of the quantum FW update (Lemma \ref{lemma:QFWLG}), we first assume that the gradient estimation and the maximum-finding are exact, with proof given in Appendix \ref{proof:lemma:qfwlg}.
Then we give the error analysis and show how to choose the parameters $\sigma_t$ in Theorem \ref{Theorem:QFWLG}, with proof given in Appendix \ref{proof:theorem:qfwlg}.

\begin{restatable}{lemma}{lemmaQFWLG}[Quantum FW update over latent group norm ball]
\label{lemma:QFWLG}
Let $\norm{\cdot}_{\+G}$ be a latent group norm corresponding to $\+G = \{\*g_1, \*g_2, \ldots,\*g_{\abs{\+G}}\}$, and let $\abs{\*g}_{\max} = \max_j \abs{\*g_j}$. 
Then, there exists a quantum algorithm computing the Frank-Wolfe update $s^* := \argmax_{\hat s\in \norm{\cdot}_{\+G}\text{-Ball}} \langle \hat{s}^\top g(\bx) \rangle$ in $O(\sqrt{\abs{\+G}}\abs{\*g}_{\max})$ calls to $\bU_f$.
\end{restatable}

\begin{restatable}{theorem}{TheoremQFWLG}[Quantum FW over latent group norm ball]
\label{Theorem:QFWLG}
By setting $\sigma_t = \frac{C_f}{\sqrt{d}L (t+2) \max_{i\in [\abs{\+G}]}|\*g_i|^{1/p_i}}$ for $t\in[T]$,  the quantum algorithm (Algorithm \ref{alg:q_fwlg}) solves the latent group norm constraint optimization problem for any precision $\varepsilon$ such that $f(\bx^{T})-f(\bx^*)\le\varepsilon$ in $T=\frac{4C_f}{\varepsilon}-2$ rounds, succeed with probability $1-p$, with $O\left(\sqrt{\abs{\+G}}\abs{\*g}_{\max} \log{\frac{C_f}{p\varepsilon}} \right)$ calls to the function value oracle $\bU_f$ per round.
\end{restatable}

\subsection{Notations and Assumptions for Quantum Computation}
\label{sec:quantumBase}

\textbf{Basic Notions in Quantum Computing.} 
Quantum computing utilizes Dirac notation as its mathematical foundation. 
Let $\{\ket{i}\}_{i=0}^{d-1}$ denote the computational basis of $\mathbb{C}^d$ as $\{\ket{i}\}_{i=0}^{d-1}$, where $\ket{i}$ is a $d$-dimensional unit vector with $1$ at the $i^{th}$ position and $0$ elsewhere. 
A $d$-dimensional quantum state is represented as a unit vector $\ket{v}=(v_1,v_2,\dots,v_d)^T=\sum_i v_i\ket{i} \in \mathbb{C}^d$ with complex amplitudes $v_i$ satisfying  $\sum_i \abs{v_i}^2=1$. 

\textbf{Composite Systems.}
The joint state of two quantum systems $\ket{v}\in \mathbb{C}^{d_1}$ and $\ket{u} \in \mathbb{C}^{d_2}$ is described by the tensor product $\ket{v} \otimes \ket{u} = (v_1u_1,v_1u_2,\dots.v_2u_1,\dots,v_{d_1}u_{d_2})\in \mathbb{C}^{d_1\times d_2}$
The $\otimes$ symbol is omitted when context permits.

\textbf{Quantum Dynamics.}
Closed system evolution is described by unitary transformations.
Quantum measurement in the computational basis probabilistically projects the state onto a basis vector $\ket{i}$ with the probability of the square of the magnitude of its amplitude.
For example, measuring $\ket{v}=\sum_i v_i\ket{i}$ yields outcome $i$ with probability $\abs{v_i}^2$, followed by post-measurement state $\ket{i}$.

\textbf{Quantum Access Models.} In general, In quantum computing, access to the objective function is facilitated through quantum oracles $Q_f$, which is a unitary transformation that maps a quantum state $\ket{x}\ket{q}$ to the state $\ket{x}\ket{q+ f(x)}$, where $\ket{x}$, $\ket{q}$ and $\ket{q+ f(x)}$ are basis states corresponding to the floating-point representations of $x$, $q$ and $q+f(x)$. 
Moreover, given the superposition input $\sum_{x,q}\alpha_{x,q}\ket{x}\ket{q}$,  by linearity the quantum oracle will output the state  $\sum_{x,q}\alpha_{x,q}\ket{x}\ket{q+ f(x)}$.

\subsection{Extended Related Works}
\label{sec:relatedWorks}

The Frank-Wolfe (FW) algorithm, also known as the conditional gradient method, has evolved through several key theoretical and applied research phases. The original FW framework \cite{frank1956algorithm} established a projection-free method for quadratic programming with optimal convergence rates when solutions lie on the feasible set boundary, a property later rigorously proven by \cite{canon1968tight}. 
Wolfe's away-step modification \cite{wolfe1970convergence} addressed boundary solution limitations, while Dunn's extension \cite{dunn1978conditional} generalized FW to smooth optimization over Banach spaces using linear minimization oracles.

Modern convergence analyzes were unified by \cite{jaggi2013revisiting}, who introduced duality gap certificates for primal-dual convergence in constrained convex optimization. 
For strongly convex objectives, \cite{garber2016linearly} demonstrated accelerated linear convergence rates. 
Projection-free optimization on non-smooth objective functions was studied in \cite{lan2013complexity,argyriou2014hybrid,pierucci2014smoothing}.
Data-dependent convergence bounds on spectahedrons were improved by \cite{garber2016faster} and \cite{allen2017linear}.

Note that the framework was extended to online and stochastic optimizations, inspiring a series of seminal contributions \cite{hazan2012projection,garber2016linearly,levy2019projection,lan2016conditional,hazan2016variance,chen2018projection,hassani2020stochastic,xie2020efficient,yurtsever2019conditional,zhang2020one}.
Our future research will explore quantum-enhanced acceleration for these online/stochastic settings.
Meanwhile, in recent years, FW methods have gained attention for their effectiveness in dealing with structured constraint problem arising in machine learning and data science, such as LASSO, SVM training, matrix completion and clustering detection. Readers are referred to \cite{bomze2021frank,pokutta2023frank} for more information.


The algorithms we develop in the matrix domain belong to the quantum algorithmic family for linear systems. 
This family originated with the seminal HHL algorithm \cite{Harrow2009}, which solves quantum linear systems and achieves exponential speedups over classical methods for well-conditioned sparse matrices.  
Subsequent improvements reduced dependency on condition number and sparsity \cite{ambainis2012variable(b),childs2017quantum,wossnig2018quantum}.
The HHL framework has been successfully adapted to machine learning tasks including support vector machines \cite{rebentrost2014quantum}, supervised and unsupervised machine learning \cite{lloyd2013quantum}, principal component analysis \cite{lloyd2014quantum} and recommendation systems \cite{kerenidis2017quantum}. 
One can reduce the condition number by preprocessing the matrix itself, and QRAM can help to accelerate such preprocessing.
Based on this, the quantum singular value estimation method was developed in \cite{kerenidis2017quantum} and was generalized in \cite{kerenidis2020quantum}. 
Furthermore, recent work integrates QSVE with state-vector tomography, amplitude amplification/estimation, and spectral norm analysis to enable top-$k$ singular vector extraction \cite{bellante2022quantum}.

\subsection{Discussion of the two Quantum Frank-Wolfe Algorithms for the Matrix Case}
\label{sec:discussion}
Note that the rank of the gradient matrix satisfies $r\le d$, Specifically, 
if the gradient matrix is low-rank ($r \ll d$), then the quantum FW algorithm with quantum top singular vector extraction (described in Subsection \ref{subsec:qfwQSVE}) achieves greater quantum speedup than the quantum power method in Subsection \ref{subsec:qfwQPM}.
Conversely, for high-rank gradient matrices, the latter algorithm (QPM) is preferable, delivering a quantum speedup of at least $O(\sqrt{d})$.

Furthermore, the repetition steps required for quantum state tomography can be parallelized in the quantum computing cluster.
By utilizing $O(d)$ quantum computers simultaneously, the dependence of $d$ in time complexity can be eliminated, giving a parallel time complexity of $\tilde{O}\left( \frac{r \sigma_1^3(M)}{(\sigma_1(M)-\sigma_2(M)) \varepsilon^2} \right)$ and $\tilde{O}\left(\frac{\sqrt{r}\sigma_1^4(M)}{(1-\sigma_1(M)){\gamma'}_{\min}^3\varepsilon^3}\right)$.

\begin{remark}
    \label{remark:square}
    Note that in this section, for simplicity of presentation, we focus on square matrices. However, all of the quantum techniques mentioned above can also be applied to non-square matrices, since the quantum singular value estimation can be applied to non-square matrices\cite{kerenidis2020quantum}.
\end{remark}

\begin{remark}
\label{remark:noG}
Note that both the classical and quantum algorithms in this section assume that the gradients are pre-stored at the memory. 
In some applications, obtaining the gradients may not be easy, and even directly loading them into the memory would scale linearly with the size of the matrix. 
This work focuses only on the computation of the update direction, but the gradient calculation time, which is also ignored in classical algorithms \cite{jaggi2013revisiting}, is explicitly included in the result Table \ref{table:resultM}. 
This is because in quantum computing, there exist several well-established algorithms for gradient estimation \cite{Jordan05, Gilyen2019}, and the potential acceleration in the gradient calculation part is left for future exploration.
\end{remark}

\section{Proof Detail}
\label{lab:SecProof}

\subsection{Proof of Lemma \ref{Lemma:QGC}}
\label{proof:Lemma:QGC}
\LemmaQGC*

\begin{proof}

By choosing appropriate $\sigma$, we now construct a gradient unitary $\bU_g:\ket{i}\ket{\bx}\ket{0} \rightarrow \ket{i}\ket{\bx}\ket{g_i(\bx)}$ as follows:
\begin{align}
    &\ket{i}\ket{\bx}\ket{0}\ket{0}\ket{0}\ket{0} \notag\\
    &\rightarrow \ket{i}\ket{\bx}\ket{\bx+\sigma\be_i}\ket{0}\ket{0}\ket{0} \label{line:add_sigma}\\
    &\rightarrow \ket{i}\ket{\bx}\ket{\bx+\sigma\be_i}\ket{f(\bx+\sigma\be_i)}\ket{f(\bx)}\ket{0} \label{line:apply_f}\\\
    &\rightarrow \ket{i}\ket{\bx}\ket{\bx+\sigma\be_i}\ket{f(\bx+\sigma\be_i)}\ket{f(\bx)}\ket{\frac{f(\bx+\sigma \be_i)-f(\bx)}{\sigma}} \label{line:add_divide}\\
    &\rightarrow \ket{i}\ket{\bx}\ket{\frac{f(\bx+\sigma \be_i)-f(\bx)}{\sigma}} \label{line:uncompute}\\
    &=\ket{i}\ket{\bx}\ket{g_i(\bx)},
\end{align}
where \cref{line:add_sigma} is by adding $\sigma$ at the $i$-th entry of the third register, \cref{line:apply_f} is by applying $\bU_f$ based on the second and the third register, \cref{line:add_divide} is by applying addition and division based on the fourth and the fifth register, \cref{line:uncompute} is by uncomputing the third, fourth and fifth register. 
For the complexity, this $\bU_g$ takes two queries of $\bU_f$ and $O(1)$ elementary gates to get the approximate gradient.

\end{proof}

\subsection{Proof of Lemma \ref{Lemma:QMF}}
\label{proof:Lemma:QMF}
\LemmaQMF*
\begin{proof}
    We restate the quantum minimum finding algorithm here for reader benefits \cite{durr1996quantum}:
        Choose threshold index $0\leq j \leq d-1 $ uniformly at random.
        Repeat the following and return $j$ when the total running time is more than $22.5\sqrt{d}+1.4\log (d)$:
        \begin{itemize}
            \item[1.] Prepare the state $\sum_i^d \ket{i}\ket{x}\ket{g_i(x)}\ket{0}$.
            \item[2.] Set the third register to $\ket{1}$ conditioned on the value of the second register smaller than $g_j(x)$
            \item[3.] Apply the quantum exponential Grover search algorithm for the third register being $1$.
            \item[4.] Measure the first register in computation basis, if the measurement result is smaller than $g_j(x)$, set $j$ to be the measurement result.
        \end{itemize}

        By Theorem $1$ of \cite{durr1996quantum}, the algorithm finds the minimum $g_i(x)$ with probability $1/2$, $O(\sqrt{d})$ applications of $\bU_g$, $\bU_g^{\dagger}$ and $O(\sqrt{d})$ elementary gates. The probability can be boost to $1-\delta$ with $O(\log (1/\delta))$ repeats and taking the minimum of the outputs.
        
        This algorithm can be modified into the quantum maximum absolute value finding algorithm by setting the third register to $\ket{1}$ conditioned on the value of the second register greater than $\abs{g_j(x)}$ in Step $2$, and set $j$ to be the measurement result that is greater than $\abs{g_j(x)}$ in Step $4$.

        However, with the estimated error, the greatest estimated gradient component $g_{max}(\bx)$ may not have the same index of $\nabla f_{max}(\bx)$. 
        As $|g_i(\bx)-\nabla f_i(\bx)|\le \epsilon$ for each $i$, in the worst case, there exists $i$ such that $\abs{g_i(\bx)}=\abs{\nabla f_i(\bx)}+\epsilon \geq \abs{g_{i^*}(\bx)}= \max_{j \in [d]}\abs{\nabla f_j(\bx)}-\epsilon$, the maximum finding algorithm will give such $g_i(\bx)$ as outcome, which is greater than $\max_{j \in [d]}\abs{\nabla f_j(\bx)}-2\epsilon$.
        
        Similarly, As $|g_i(\bx)-\nabla f_i(\bx)|\le \epsilon$ for each $i$, in the worst case, there exists $i$ such that $\abs{g_i(\bx)}=\abs{\nabla f_i(\bx)}-\epsilon \le \abs{g_{i^*}(\bx)}= \min_{j \in [d]}\abs{\nabla f_j(\bx)}+\epsilon$, the minimum finding algorithm will give such $g_i(\bx)$ as outcome, which is less than $\min_{j \in [d]}\abs{\nabla f_j(\bx)}+2\epsilon$.
        Similar proof processes can be employed to derive the error bounds for the minimum/maximum search, which gives the lemma.
\end{proof}

\subsection{Proof of Theorem \ref{Theorem:QFWSC}}
\label{proof:Theorem:QFWSC}
\TheoremQFWSC*
\begin{proof}

    By Lemma \ref{Lemma:FDM} and the inequality between $\ell_2$ norm and $\ell_{\infty}$ norm, we have
    \begin{equation}
        |g_i(\bx)-\nabla f_i(\bx)|\le \norm{g(\bx)-\nabla f(\bx)}_{\infty} \le \norm{g(\bx)-\nabla f(\bx)}_2 \le \frac{\sqrt{d}L\sigma}{2}.
    \end{equation}

    By Lemma \ref{Lemma:QMF}, after the quantum approximate maximum absolute value finding, we have an estimated maximum gradient component which satisfied
    \begin{equation}
        \abs{\nabla f_{i^*}(\bx)}\ge \max_{j \in [d]}\abs{\nabla f_j(\bx)}-\sqrt{d}L\sigma
    \end{equation}

    Set $\bs=-\be_{i^*}$, we have
    \begin{align}
        \label{equ:1T1}
        \langle{\bs, \nabla f(\bx^{(t)})}\rangle & = -\abs{\nabla f_{i^*}(\bx^{(t)})} \nonumber \\ & \le -\max_{j \in [d]}\abs{\nabla f_j(\bx^{(t)})}+\sqrt{d}L\sigma_t \nonumber \\ & = -\langle{\be_{\argmax_{i \in [d]}|\nabla_if(\bx^{(t)})|}, \nabla f(\bx^{(t)})\rangle}+\sqrt{d}L\sigma_t \nonumber \\ & = \min_{\hat{\bs} \in \cD}\langle{\hat{\bs}, \nabla f(\bx^{(t)})\rangle}+\sqrt{d}L\sigma_t.
    \end{align}

    By the update rule and the definition of the curvature, we have
    \begin{equation}
        \label{equ:1T2}
        f(\bx^{(t+1)})=f((1-\gamma_t) \bx^{(t)}+\gamma_t \bs) \le f(\bx^{(t)})+\gamma_t \langle{\bs-\bx^{(t)}, \nabla f(\bx^{(t)})}\rangle + \frac{\gamma_t^2}{2}C_f
    \end{equation}

    Combining Inequality \ref{equ:1T1} and \ref{equ:1T2}, we have
    \begin{equation}
        f(\bx^{(t+1)}) \le f(\bx^{(t)})+\gamma_t (\min_{\hat{\bs} \in \cD}\langle{\hat{\bs}, \nabla f(\bx)\rangle}- \langle{\bx^{(t)}, \nabla f(\bx^{(t)})}\rangle) +\sqrt{d}\gamma_t L\sigma_t + \frac{\gamma_t^2}{2}C_f.
    \end{equation}

    Let $h(\bx^{(t)}):=f(\bx^{(t)})-f(x^*)$, we have
    \begin{align}
        h(\bx^{(t+1)}) & \le h(\bx^{(t)})+\gamma_t (\min_{\hat{\bs} \in \cD}\langle{\hat{\bs}, \nabla f(\bx)\rangle}- \langle{\bx^{(t)}, \nabla f(\bx^{(t)})}\rangle) +\sqrt{d}\gamma_t L\sigma_t + \frac{\gamma_t^2}{2}C_f \nonumber \\ & \le h(\bx^{(t)})-\gamma_t h(\bx^{(t)}) +\sqrt{d}\gamma_t L\sigma_t + \frac{\gamma_t^2}{2}C_f \nonumber \\ & =  (1-\gamma_t) h(\bx^{(t)}) +\sqrt{d}\gamma_t L\sigma_t + \frac{\gamma_t^2}{2}C_f.
    \end{align}

    Set $\gamma_{t}=\frac{2}{t+2},\sigma_t=\frac{\gamma_t C_f}{2\sqrt{d}L}$, we have

    \begin{equation}
        h(\bx^{(t+1)}) \le  \left(1-\frac{2}{t+2}\right) h(\bx^{(t)}) + \left(\frac{2}{t+2}\right)^2C_f.
    \end{equation}

    Using a similar induction as shown in \cite{jaggi2013revisiting} over $t$, we have
    \begin{equation}
        h(\bx^{(t)}) \le  \frac{4C_f}{t+2}.
    \end{equation}
    We will restate this induction in Lemma \ref{Lemma:induction} for reader benefit.

    Thus, set $\gamma_{t}=\frac{2}{t+2},\sigma_t=\frac{C_f}{\sqrt{d}L(t+2)}$ for all $t\in [T]$, after $T=\frac{4C_f}{\varepsilon}-2$ rounds, we have
    \begin{equation}
        f(\bx^{(T)})-f(x^*) \le  \varepsilon,
    \end{equation}
    for any $\varepsilon>0$.
    
    In each round, by Lemma \ref{Lemma:QGC}, two queries to the quantum function value oracle are needed to construct the quantum gradient oracle. 
    Then by lemma \ref{Lemma:QMF}, $O(\sqrt{d}\log{\frac{1}{\delta}})$ queries to the quantum gradient oracle are needed to find the index of the estimated maximum gradient component with successful probability of $1-\delta$.
    Since each maximum finding succeeds with probability $1-\delta$, the probability that all $T$ iterations succeed is at least $1-T\delta$. By setting $\delta=p/T$, we ensure an overall success probability of at least $1-p$.
    Therefore, $O\left(\sqrt{d} \log{\frac{C_f}{p\varepsilon}} \right)$ queries to the quantum function value oracle are needed in each iteration.
    Then the theorem follows.
\end{proof}

We restate the proof of the induction we use in Theorem \ref{Theorem:QFWSC} for reader benefit.
\begin{lemma}(\cite{jaggi2013revisiting})
    \label{Lemma:induction}
    If for any $t\in[N]$,
    \begin{equation}
        h(\bx^{(t+1)}) \le  \left(1-\frac{2}{t+2}\right) h(\bx^{(t)}) + \left(\frac{2}{t+2}\right)^2C_f,
    \end{equation}
    then
    \begin{equation}
        h(\bx^{(t)}) \le  \frac{4C_f}{t+2}.
    \end{equation}
\end{lemma}
\begin{proof}
    For $t=0$, we have
    \begin{equation}
        h(\bx^{(1)}) \le  \left(1-\frac{2}{0+2}\right) h(\bx^{(0)}) + \left(\frac{2}{0+2}\right)^2C_f = C_f.
    \end{equation}
    Assume that $h(\bx^{(t)}) \le  \frac{4C_f}{t+2}$, we have
    \begin{align}
        h(\bx^{(t+1)}) & \le  \left(1-\frac{2}{t+2}\right) h(\bx^{(t)}) + \left(\frac{2}{t+2}\right)^2C_f \nonumber \\
        & \le \left(1-\frac{2}{t+2}\right) \frac{4C_f}{t+2} + \left(\frac{2}{t+2}\right)^2C_f \nonumber \\
        & = \left(1-\frac{1}{t+2}\right) \frac{4C_f}{t+2} + \left(\frac{2}{t+2}\right)^2C_f \nonumber \\
        & = \frac{t+1}{t+2} \frac{4C_f}{t+2} \le \frac{t+2}{t+3} \frac{4C_f}{t+2} =\frac{4C_f}{t+3},
    \end{align}
    which gives the lemma.
\end{proof}

\subsection{Proof of Lemma \ref{QGB}}
\label{proof:QGE}

The framework of quantum gradient estimator originates from Jordan quantum gradient estimation method \cite{Jordan05}, but Jordan algorithm did not give any error bound because the analysis of it was given by omitting the high-order terms of Taylor expansion of the function directly. 
In 2019, the quantum gradient estimation method with error analysis was given in \cite{Gilyen2019}, and was applied to the general convex optimization problem \cite{van2020convex,chakrabarti2020quantum}.
In those case, however, $O(\log{n})$ repetitions were needed to estimate the gradient within an acceptable error. 
The query complexity was then improved to $O(1)$ in \cite{he2022quantum, he2024quantum}. 
Here we use the version of \cite{he2024quantum} (Algorithm \ref{Alg:JQGE}). 

\begin{algorithm}[H]
    \caption{Bounded-error Jordan quantum gradient estimation \cite{he2024quantum}}
    \label{Alg:JQGE}
    \begin{algorithmic}[1]
        \State {\textbf{Input:}} point $x$, parameters $r, \rho, \epsilon$.
        \State {\textbf{Output:}} $g(x)$
            \State Prepare the initial state: $d$ $b$-qubit registers $\ket{0^{\otimes b},0^{\otimes b},\dots,0^{\otimes b}}$ where $b=\log_2 \cfrac{G\rho}{4\pi d^2 \beta r}$. 
            Prepare $1$ $c$-qubit register $\ket{0^{\otimes c}}$ where $c=\log_2{\cfrac{16\pi d}{\rho}}-1$.
            And prepare $\ket{y_0}=\cfrac{1}{\sqrt{2^d}}\sum_{a\in\{0,1,\dots,2^d-1\}}e^{\cfrac{2\pi i a}{2^d}}\ket{a}$.
            \State Apply Hadamard transform to the first $d$ registers.
            \State Perform the quantum query oracle $Q_{F}$ to the first $d+1$ registers, where  $F(u)=\cfrac{2^b}{2Gr} \left [f \left (x+\cfrac{r}{2^b} \left (u-\cfrac{2^b}{2}\mathbbm{1} \right ) \right )-f(x) \right ]$, and the result is stored in the $(d+1)$th register.
            \State Perform the addition modulo $2^c$ operation to the last two registers.
            \State Apply the inverse evaluating oracle $Q_{F}^{-1}$ to the first $d+1$ registers.
            \State Perform quantum inverse Fourier transformations to the first $d$ registers separately.
            \State Measure the first $d$ registers in computation bases respectively to get $m_1,m_2,\dots,m_n$.
            \State $g(x)=\widetilde{\nabla} f(x)=\cfrac{2G}{2^b} \left (m_1-\cfrac{2^b}{2},m_2-\cfrac{2^b}{2},\dots, m_n-\cfrac{2^b}{2} \right )^{\mathrm{T}}$.
    \end{algorithmic}
\end{algorithm}

\QGB*

\begin{proof}
    The primary additional gate overhead originates from the quantum Fourier transformation (QFT). Each QFT requires  $O(\log d)$ elementary gates, and for $d$ such operations, the total additional elementary gate overhead is $O(d\log d)$. Consequently, the additional elementary gate overhead is $O(d\log d)$.
    
    The states after Step $3$ will be:
    \begin{align}
        \frac{1}{\sqrt{2^n}}\sum_{a\in\{0,1,\dots,2^n-1\}}e^{\frac{2\pi i a}{2^n}} \ket{0^{\otimes b},0^{\otimes b},\dots,0^{\otimes b}} \ket{0^{\otimes c}} \ket{a}.
    \end{align}
    After Step $4$:
    \begin{align}
       \frac{1}{\sqrt{2^{bn+c}}}\sum_{u_1,u_2,\dots,u_n\in\{0,1,\dots,2^b-1\}} 
       \sum_{a\in\{0,1,\dots,2^c-1\}} 
       e^{\frac{2\pi i a}{2^n}} 
       \ket{u_1,u_2,\dots,u_n} 
       \ket{0^{\otimes c}} \ket{a}. 
    \end{align}
    After Step $5$:
    \begin{align}
       \frac{1}{\sqrt{2^{bn+c}}}\sum_{u_1,u_2,\dots,u_n\in\{0,1,\dots,2^b-1\}} 
       \sum_{a\in\{0,1,\dots,2^c-1\}} 
       e^{\frac{2\pi i a}{2^n}}   
       \ket{u_1,u_2,\dots,u_n} 
       \ket{F(u)} \ket{a}. 
    \end{align}
    After Step $6$:
    \begin{align}
       \frac{1}{\sqrt{2^{bn+c}}}\sum_{u_1,u_2,\dots,u_n\in\{0,1,\dots,2^b-1\}} 
       \sum_{a\in\{0,1,\dots,2^c-1\}} 
       e^{2\pi i F(u)}
       e^{\frac{2\pi i a}{2^n}} 
       \ket{u_1,u_2,\dots,u_n} 
       \ket{F(u)} \ket{a}. 
    \end{align}
    After Step $7$:
    \begin{align}
       \frac{1}{\sqrt{2^{bn+c}}}\sum_{u_1,u_2,\dots,u_n\in\{0,1,\dots,2^b-1\}} 
       \sum_{a\in\{0,1,\dots,2^c-1\}} 
       e^{2\pi i F(u)}
       e^{\frac{2\pi i a}{2^n}} 
       \ket{u_1,u_2,\dots,u_n} 
       \ket{0^{\otimes c}} \ket{a}. 
    \end{align}
    In the following, the last two registers will be omitted:
    \begin{align}
       \frac{1}{\sqrt{2^{bn}}}\sum_{u_1,u_2,\dots,u_n\in\{0,1,\dots,2^b-1\}} 
       e^{2\pi i F(u)}
       \ket{u_1,u_2,\dots,u_n} .
    \end{align}
    And then we simply relabel the state by changing $u \to v=u-\frac{2^b}{2}$:
    \begin{align}
        \label{phi}
        \frac{1}{\sqrt{2^{bn}}}\sum_{v_1,v_2,\dots,v_n\in\{-2^{b-1},-2^{b-1}+1,\dots,2^{b-1}\}} 
        e^{2\pi i F(v)}
        \ket{v}. 
    \end{align}
    We denote Formula (\ref{phi}) as $\ket{\phi}$. Let $g=\nabla f(x)$, and consider the idealized state
    \begin{align}
        \label{psi}
        \ket{\psi}=\frac{1}{\sqrt{2^{bn}}}\sum_{v_1,v_2,\dots,v_n\in\{-2^{b-1},-2^{b-1}+1,\dots,2^{b-1}\}} 
        e^{\frac{2\pi i g \cdot v}{2G}}
        \ket{v}.
    \end{align}
    After Step $9$, from the analysis of phase estimation \cite{Brassard00}:
    \begin{align}
        \Pr[\abs{\frac{Ng_i}{2G}-m_i}>e]<\frac{1}{2(e-1)}, \forall{i\in [n]}.
    \end{align}
    Let $e=n/\rho +1$, where $1\geq \rho>0$. We have
    \begin{align}
        \Pr[\abs{\frac{Ng_i}{2G}-m_i}>n/\rho +1]<\frac{\rho}{2n}, \forall{i\in [n]}.
    \end{align}
    Note that the difference in the probabilities of measurement on $\ket{\phi}$ and $\ket{\psi}$ can be bounded by the trace distance between the two density matrices:
    \begin{align}
        \|\dyad{\phi}{\phi}-\dyad{\psi}{\psi}\|_1
        =2\sqrt{1-|\braket{\phi}{\psi}|^2}
        \leq 2\|\ket{\phi}-\ket{\psi}\|.
    \end{align}
    Since $f$ is $L$-smooth, we have
    \begin{align}
        F(v) & \leq \frac{2^b}{2Gr}[f(x+\frac{r v}{N})-f(x)]+\frac{1}{2^{c+1}} \nonumber\\
        & \leq \frac{2^b}{2Gr}[\frac{r}{2^b} g \cdot v +\frac{L (r v)^2}{2^{2b}}]
        +\frac{1}{2^{c+1}} \nonumber\\
        & \leq \frac{g \cdot v}{2G}+\frac{2^b L r n}{4G}+\frac{1}{2^{c+1}}.
    \end{align}
    Then,
    \begin{align}
        \|\ket{\phi}-\ket{\psi}\|^2 & = \frac{1}{2^{bn}} \sum_v | e^{2 \pi i F(v)} - e^{\frac{2 \pi i g \cdot v}{2G}} |^2 \nonumber \\
        & \leq \frac{1}{2^{bn}} \sum_v | 2 \pi i F(v) - \frac{2 \pi i g \cdot v}{2G} |^2 \nonumber \\
        & \leq \frac{1}{2^{bn}} \sum_v 4 \pi^2 (\frac{2^b L r n}{4G}+\frac{1}{2^{c+1}})^2.
    \end{align}
    Set $b=\log_2 \frac{G\rho}{4\pi n^2 L r}$, $c=\log_2{\frac{4G}{2^b n L r }}-1$.  We have
    \begin{align}
        \|\ket{\phi}-\ket{\psi}\|^2 \leq \frac{\rho^2}{16n^2},
    \end{align}
    which implies $\|\dyad{\phi}{\phi}-\dyad{\psi}{\psi}\|_1 \leq \frac{\rho}{2n}$. Therefore, by the union bound,
    \begin{align}
        \Pr[\abs{\frac{2^b g_i}{2G}-m_i}>n/\rho +1]<\frac{\rho}{n}, \forall{i\in [n]}.
    \end{align}
    Furthermore, there is 
    \begin{align}
        \Pr[\abs{g_i-\widetilde{\nabla}_i f(x)}>\frac{2G(n/\rho +1)}{2^b}]<\frac{\rho}{n}, \forall{i\in [n]},
    \end{align}
    as $b=\log_2 \frac{G\rho}{4\pi n^2 L r}$, we have
    \begin{align}
        \Pr[\abs{g_i-\widetilde{\nabla}_i f(x)}>8 \pi n^2 (n/\rho +1) L r/\rho]<\frac{\rho}{n}, \forall{i\in [n]}.
    \end{align}
    By the union bound, we have
    \begin{align}
        \Pr[\|{g-\widetilde{\nabla} f(x)}\|_{\infty}>8 \pi n^2 (n/\rho +1) L r/\rho]<\rho,
    \end{align}
    which gives the lemma.
\end{proof}

\subsection{Proof of Theorem \ref{Theorem:QFWJ}}
\label{proof:Theorem:QFWJ}
\TheoremQFWJ*

\begin{proof}
    By Lemma \ref{QGB}, with probability greater than $\rho$, we have
    \begin{equation}
        |g_i(\bx)-\nabla f_i(\bx)|\le \norm{g(\bx)-\nabla f(\bx)}_{\infty} \le 8 \pi d^2 (d/\rho +1) L r/\rho.
    \end{equation}

    Then the maximum component's coordinate of the estimated gradient  $i^*= \argmax_{i \in [d]}|g_i(\bx^{(t)})|$ satisfies
    \begin{equation}
        \abs{\nabla f_{i^*}(\bx)}\ge \max_{j \in [d]}\abs{\nabla f_j(\bx)}-16 \pi d^2 (d/\rho +1) L r/\rho
    \end{equation}

     Set $\bs=-\be_{i^*}$, we have
    \begin{align}
        \label{equ:3T1}
        \langle{\bs, \nabla f(\bx^{(t)})}\rangle & = -\abs{\nabla f_{i^*}(\bx^{(t)})} \nonumber \\ & \le -\max_{j \in [d]}\abs{\nabla f_j(\bx^{(t)})}+16 \pi d^2 (d/\rho +1) L r/\rho \nonumber \\ & = -\langle{\be_{\argmax_{i \in [d]}|\nabla_if(\bx^{(t)})|}, \nabla f(\bx^{(t)})\rangle}+16 \pi d^2 (d/\rho +1) L r/\rho \nonumber \\ & = \min_{\hat{\bs} \in \cD}\langle{\hat{\bs}, \nabla f(\bx^{(t)})\rangle}+16 \pi d^2 (d/\rho +1) L r/\rho.
    \end{align}

    By the update rule and the definition of the curvature, we have
    \begin{equation}
        \label{equ:3T2}
        f(\bx^{(t+1)})=f((1-\gamma_t) \bx^{(t)}+\gamma_t \bs) \le f(\bx^{(t)})+\gamma_t \langle{\bs-\bx^{(t)}, \nabla f(\bx^{(t)})}\rangle + \frac{\gamma_t^2}{2}C_f
    \end{equation}

    Combining Inequality \ref{equ:3T1} and \ref{equ:3T2}, we have
    \begin{equation}
        f(\bx^{(t+1)}) \le f(\bx^{(t)})+\gamma_t (\min_{\hat{\bs} \in \cD}\langle{\hat{\bs}, \nabla f(\bx)\rangle}- \langle{\bx^{(t)}, \nabla f(\bx^{(t)})}\rangle) + 16 \pi d^2 (d/\rho +1) L \gamma_t r/\rho + \frac{\gamma_t^2}{2}C_f.
    \end{equation}

    Let $h(\bx^{(t)}):=f(\bx^{(t)})-f(x^*)$, we have
    \begin{align}
        h(\bx^{(t+1)}) & \le h(\bx^{(t)})+\gamma_t (\min_{\hat{\bs} \in \cD}\langle{\hat{\bs}, \nabla f(\bx)\rangle}- \langle{\bx^{(t)}, \nabla f(\bx^{(t)})}\rangle) +16 \pi d^2 (d/\rho +1) L \gamma_t r/\rho + \frac{\gamma_t^2}{2}C_f \nonumber \\
        & \le h(\bx^{(t)})-\gamma_t h(\bx^{(t)}) +16 \pi d^2 (d/\rho +1) L \gamma_t r/\rho + \frac{\gamma_t^2}{2}C_f \nonumber \\
        & =  (1-\gamma_t) h(\bx^{(t)}) +16 \pi d^2 (d/\rho +1) L \gamma_t r/\rho + \frac{\gamma_t^2}{2}C_f.
    \end{align}

    Set $\gamma_{t}=\frac{2}{t+2},r_t=\frac{\rho \gamma_t C_f}{32 \pi d^2 (d/\rho +1) L}$, we have

    \begin{equation}
        h(\bx^{(t+1)}) \le  \left(1-\frac{2}{t+2}\right) h(\bx^{(t)}) + \left(\frac{2}{t+2}\right)^2C_f.
    \end{equation}

    Using a similar induction as shown in \cite{jaggi2013revisiting} over $t$, we have
    \begin{equation}
        h(\bx^{(t)}) \le  \frac{4C_f}{t+2}.
    \end{equation}

    Thus, set $\gamma_{t}=\frac{2}{t+2},r_t=\frac{\rho  C_f}{16 \pi d^2 (d/\rho +1) L (t+2)}$ for all $t\in [T]$, after $T=\frac{4C_f}{\varepsilon}-2$ rounds, we have
    \begin{equation}
        f(\bx^{(T)})-f(x^*) \le  \varepsilon,
    \end{equation}
    for any $\varepsilon>0$.
    
    In each round, by Lemma \ref{QGB}, $O(1)$ queries to the quantum function value oracle are needed to get the estimated gradient vector. 
    Subsequent steps no longer require queries to the oracle.
    Therefore, in each round, $O(1)$ queries to the quantum function value oracle are needed.
    Then the theorem follows.

\end{proof}

\subsection{Proof of Lemma \ref{lemma:QFWLG}}
\label{proof:lemma:qfwlg}
\lemmaQFWLG*

\begin{proof}
Assume that all $\norm{\cdot}_{\*g}$ are $\ell_p$-norms, i.e. $\norm{\cdot}_{\*g_i} = \norm{\cdot}_{p_i}$ for some  ($p_i\in[1,\infty]$), and have quantum access to each $\*g_i = \{\*g_{i,1}, \*g_{i,2},\ldots, \*g_{i,\abs{\*g_i}}\}\subseteq[d]$ that load $\*g_i$ into quantum registers via
\eq{
U_{\+G}\ket{i}_A \ket{0}&\ra \ket{i}_A\ket{\*g_{i,1}}\ket{\*g_{i,2}}\ldots\ket{\*g_{i,\abs{\*g_i}}}
}
where $A$ is a $\log \abs{\+G}$ qubit register. 
For each $\ket{\*g_{i,j}}$ one can compute an approximation $\ket{g_{{\*g_{i,j}}}(\bx)}$ to the ${\*g_{i,j}}$-th component of the gradient at $\bx$ by the method in Sec.~\ref{subsec:QFWSC}. 

Noting that $\max_{\hat{\bs}\in\norm{\cdot}_p\text{-Ball}} \bs^\top \by := \norm{\by}_p^*$ and that
\eq{
\bs^* &:= \argmax_{\hat{\bs}\in \norm{\cdot}_p\text{-ball} } \bs^\top \by
}
has components
\eq{
 \bs^*_i &\propto \mathsf{sgn}(\by_i)\abs{\by_i}^{q-1}
}
where $\frac{1}{p}+\frac{1}{q}=1$, one can compute
\begin{align}
& \ket{i}_A\bigotimes_{j=1}^{\abs{\*g_i}}\ket{\*g_{i,j}} \ket{0}\ket{0}\ket{0}\ket{0} \nonumber \\
\ra & \ket{i}_A\bigotimes_{j=1}^{\abs{\*g_i}}\ket{\*g_{i,j}}\ket{g_{\*g_{i,j}}(\bx)} \ket{0}\ket{0}\ket{0} \nonumber\\
\ra &\ket{i}_A\bigotimes_{j=1}^{\abs{\*g_i}}\ket{\*g_{i,j}}\ket{g_{\*g_{i,j}}(\bx)}\ket{\mathsf{sgn}(g_{\*g_{i,j}}(\bx))\abs{g_{\*g_{i,j}}(\bx)}^{q_i-1}}\ket{0}\ket{0}\nonumber\\
\ra  & \ket{i}_A \lp \bigotimes_{j=1}^{\abs{\*g_i}}\ket{\*g_{i,j}}\ket{g_{\*g_{i,j}}(\bx)}\ket{\mathsf{sgn}(g_{\*g_{i,j}}(\bx))\abs{g_{\*g_{i,j}}(\bx)}^{q_i-1}}\rp \ket{\norm{g(\bx)_{(\*g_i)}}_{p_i}} \ket{0}\nonumber\\
\ra & \ket{i}_A \lp \bigotimes_{j=1}^{\abs{\*g_i}}\ket{\*g_{i,j}}\ket{g_{\*g_{i,j}}(\bx)}\ket{\mathsf{sgn}(g_{\*g_{i,j}}(\bx))\abs{g_{\*g_{i,j}}(\bx)}^{q_i-1}}\rp \ket{\norm{g(\bx)_{(\*g_i)}}_{p_i}} \ket{\norm{g(\bx)_{(\*g_i)}}_{p_i}^*}
\end{align}
Apply quantum maximum finding to the last register can then be used to find $s^*$ in $O(\sqrt{\abs{\+G}})$ iterations.
Each $g_{\*g_{i,j}}(\bx)$ requires $2$ queries to $U_f$, totally $O(\abs{\*g_i})$ queries for a fixed $i$.
In the above the index $i$ ranges over $i=1,2,\ldots, \abs{\+G}$. The query complexity is therefore $O(\sqrt{\abs{\+G}}\abs{\*g}_{\max})$, compared with the classical $\sum_{\.g\in \+G}\abs{\.g}$.
Then the lemma follows.
\end{proof}

\subsection{Proof of Theorem \ref{Theorem:QFWLG}}
\label{proof:theorem:qfwlg}
\TheoremQFWLG*

\begin{proof}

Let the true gradient component be $ g_{\*g_{i,j}}(\bx) $, and its estimated value be $ \tilde{g}_{\*g_{i,j}}(\bx)$ such that $\abs{\tilde{g}_{\*g_{i,j}}(\bx) - g_{\*g_{i,j}}(\bx)} \leq  \frac{\sqrt{d}L\sigma}{2} $.  
According to Step 7 of the algorithm, the dual norm computation involves:
\begin{align}
\left\| g(\bx)_{(\*g_i)} \right\|_{p_i}^* = \max_{\bs \in \mathbb{R}^{|\*g_i|}} \left\{ \sum_{j=1}^{|\*g_i|} \bs_j g_{\*g_{i,j}}(\bx) \ \bigg| \ \|\bs\|_{q_i} \leq 1 \right\},
\end{align}
where \( \frac{1}{p_i} + \frac{1}{q_i} = 1 \). The estimated dual norm is:
\begin{align}
\left\| \tilde{g}(\bx)_{(\*g_i)} \right\|_{p_i}^* = \max_{\bs \in \mathbb{R}^{|\*g_i|}} \left\{ \sum_{j=1}^{|\*g_i|} \bs_j \tilde{g}_{\*g_{i,j}}(\bx) \ \bigg| \ \|\bs\|_{q_i} \leq 1 \right\}.
\end{align}

The dual norm error can be decomposed as
\begin{align}
\left| \left\| \tilde{g}(\bx)_{(\*g_i)} \right\|_{p_i}^* - \left\| g(\bx)_{(\*g_i)} \right\|_{p_i}^* \right| \leq \max_{\|\bs\|_{q_i} \leq 1} \left| \sum_{j=1}^{|\*g_i|}  \frac{s_j\sqrt{d}L\sigma}{2} \right|.
\end{align}
By H{\"o}lder’s inequality, for any $\bs$ satisfying $\|\bs\|_{q_i} \leq 1 $, let $\delta_{(\*g_i)}$ be the vector in $\mb{R}^{\*g_i}$ with all the component being $\frac{\sqrt{d}L\sigma}{2}$ we have
\begin{align}
\left| \sum_{j=1}^{|\*g_i|}  \frac{s_j\sqrt{d}L\sigma}{2} \right| \leq \|\bs\|_{q_i} \cdot \|\delta_{(\*g_i)}\|_{p_i} \leq \|\delta_{(\*g_i)}\|_{p_i}.
\end{align}
Since $ \|\delta_{(\*g_i)}\|_{p_i} \leq \frac{\sqrt{d}L\sigma |\*g_i|^{1/p_i}}{2}$, it follows that
\begin{align}
\left| \left\| \tilde{g}(\bx)_{(\*g_i)} \right\|_{p_i}^* - \left\| g(\bx)_{(\*g_i)} \right\|_{p_i}^* \right| \leq \|\delta_{(\*g_i)}\|_{p_i} \leq \frac{\sqrt{d}L\sigma |\*g_i|^{1/p_i}}{2}.
\end{align}

Then, by Lemma \ref{Lemma:QMF} and \ref{lemma:QFWLG}, after Step 8, we have
\begin{align}
    \left\| \tilde{g}(\bx)_{(\*g_{i_t})} \right\|_{p_{i_t}}^* \geq \max_{i\in \abs{\+G}} \left\| g(\bx)_{(\*g_i)} \right\|_{p_i}^* - \sqrt{d}L\sigma \max_{i\in [\abs{\+G}]} |\*g_i|^{1/p_i},
\end{align}
succeed with probability at least $1-\delta$, with query complexity of $O\left(\sqrt{\abs{\+G}}\abs{\*g}_{\max}\log{\frac{1}{\delta}}\right)$.
Set $\delta=\frac{p}{T}$ to ensure that this procedure succeeds for all $T$ iterations.

The rest parallels the proof of Theorem \ref{Theorem:QFWSC}. 
Set $\sigma_t = \frac{C_f}{\sqrt{d}L (t+2) \max_{i\in [\abs{\+G}]}|\*g_i|^{1/p_i}}$ for all $t \in [T]$, after $T=\frac{4C_f}{\varepsilon}-2$ rounds, we have
    \begin{equation}
        f(\bx^{(T)})-f(x^*) \le  \varepsilon,
    \end{equation}
    for any $\varepsilon>0$.
    Then the theorem follows.
\end{proof}

\subsection{Proof of Lemma \ref{Lemma:QGSVE}}
\label{proof:Lemma:QGSVE}
\LemmaQGSVE*

\begin{proof}
    Initialize the quantum registers to the uniform superposition state by using Hadamard gates, we have
    \begin{align}
        H^{\otimes d}\ket{0}\ket{0}\ket{0} \to \sum_{i}^d \ket{i}\ket{0}\ket{0}.
    \end{align}
    By Assumption \ref{asmp:QMA}, we can perform the mapping
    \begin{align}
        \sum_{i}^d \ket{i}\ket{0}\ket{0} \to \frac{1}{\norm{M}_F}\sum_i^d\sum_j^d M_{ij} \ket{i}\ket{j}\ket{0},
    \end{align}
    in time $\tilde{O}(1)$. Note that
    \begin{align}
        \frac{1}{\norm{M}_F}\sum_i^d\sum_j^d M_{ij} \ket{i}\ket{j}\ket{0}=\frac{1}{\norm{M}_F}\sum_i^k \sigma_i \ket{\bu_i} \ket{\bv_i}\ket{0}.
    \end{align}
    Then by the quantum singular estimation algorithm (QSVE, Lemma \ref{Lemma:QSVE}), we have
    \begin{align}
        \frac{1}{\norm{M}_F}\sum_i^d\sum_j^d M_{ij} \ket{i}\ket{j}\ket{0} \to \frac{1}{\norm{\nabla}_F}\sum_i^k \sigma_i \ket{\bu_i} \ket{\bv_i}\ket{\overline{\sigma}_i},
    \end{align}
    with the cost of $O\left(\frac{\norm{M}_F \text{poly}\log{d}}{\epsilon}\right)$. 
    This process of generating such a state is treated as an oracle which will be invoked multiple times in the quantum maximum finding. 
    This requires that the errors in the estimates of the singular values should be consistent across multiple runs. 
    Note that the randomness of QSVE comes from the quantum phase estimation algorithm, and the QSVE algorithm of Lemma \ref{Lemma:QSVE} uses a consistent version of phase estimation.
    This consistency in phase estimation guarantees that the error patterns are reproducible, thereby maintaining uniform errors over repeated oracle calls.
    
    Set $\epsilon\leq (\sigma_1-\sigma_2)/2$ to ensure that even with the error of singular value estimation, the estimated largest singular value is still larger than the estimated second largest singular value, which can ensure that when we use the quantum maximum finding algorithm, if succeed, we will always get the superposition state corresponding to the largest singular value.  
    By Lemma \ref{Lemma:QMF}, the cost of finding the largest singular value is $O\left(\frac{1}{\sqrt{p}}\right)$. 
    By Lemma \ref{Lemma:L2T}, $O(\frac{d \log{d}}{\delta^2})$ repeats are needed to tomography the corresponding singular vectors of the largest singular value. 
    
    Therefore, the overall complexity is $O\left( \frac{\norm{M}_F d \text{poly}\log{d}}{\sqrt{p} \epsilon \delta^2} \right)$.
\end{proof}

\subsection{Proof of Theorem \ref{theorem:qfwqsvd}}
\label{proof:Theorem:QFWQSVD}
\theoremqfwqsvd*
\begin{proof}
    By Lemma \ref{Lemma:QGSVE}, set $\epsilon_t\leq (\sigma_1(M)-\sigma_2(M))/2$ to ensure that the quantum maximum finding algorithm, if succeed, will always get the superposition state of the largest singular value. 
    As the QSVE algorithm from Lemma \ref{Lemma:QSVE} use a consistent version of phase estimation, the estimated error of the singular value will keep unchanged.
    Thus, we can measure the register of singular value in the computational basic, to check whether the quantum maximum finding succeed, to boost up the success probability.
    By Lemma \ref{Lemma:QGSVE}, we obtain the estimated singular vectors $\bu, \bv$, which satisfy
    $\norm{\bu-\bu_{top}}\leq\delta_t$, $\norm{\bv-\bv_{top}}\leq\delta_t$, with time complexity $O\left( \frac{\norm{M}_F d \text{poly}\log{d}}{\sqrt{p} \epsilon \delta^2} \right)$.

    Note that in the matrix case, the linear optimization subproblem of the Frank-Wolfe framework
    \begin{equation}
        \min_{\hat{S} \in \cD}\langle{\hat{S}, M_t\rangle} \quad \text{s.t.} \quad \tr{\hat{S}}\le 1
    \end{equation}
    is equivalent to the following problem
    \begin{equation}
        \min_{\bx, \by\in \mathbb{R}^d} \ \bx^\top M_t \by \quad \text{s.t.} \quad \|\bx\|, \|\by\| \leq 1.
    \end{equation}
    Therefore, since the update direction $S=\bu^{\top}\bv $, the solution quality of the linear subproblem can be bounded with the solution quality of the equivalent problem, that is 
    \begin{align}
        \label{equ:MT0}
        \langle{S,M_t\rangle}-\min_{\hat{S} \in \cD}\langle{\hat{S}, M_t\rangle}&=\langle{\bu^{\top}\bv,M_t\rangle}-\min_{\hat{S} \in \cD}\langle{\hat{S}, M_t\rangle} \nonumber \\
        &=\bu^{\top}M_t\bv-\min_{\bx, \by\in \mathbb{R}^d} \ \bx^\top M_t \by \nonumber \\
        &= \bu^{\top}M_t\bv-\bu_{top}^{\top}M_t\bv_{top}.
    \end{align}
    
    Then by Lemma \ref{Lemma:QGSVE} and \ref{lemma:technical1}, we have
    \begin{equation}
        \label{equ:MT1}
        \abs{\bu^{\top}M_t\bv-\bu_{top}^{\top}M_t\bv_{top}}\le 2\sigma_1(M_t)\delta_t.
    \end{equation}

     By the update rule and the definition of the curvature, for each round $t$, we have
    \begin{equation}
        \label{equ:MT2}
        f(X^{(t+1)})=f((1-\gamma_t) X^{(t)}+\gamma_t S) \le f(X^{(t)})+\gamma_t \langle{S-X^{(t)}, M_t}\rangle + \frac{\gamma_t^2}{2}C_f
    \end{equation}

    Combining Inequality \ref{equ:MT0}, \ref{equ:MT1} and \ref{equ:MT2}, we have
    \begin{equation}
        f(X^{(t+1)}) \le f(X^{(t)})+\gamma_t (\min_{\hat{S} \in \cD}\langle{\hat{S}, M_t\rangle}- \langle{X^{(t)}, M_t)}\rangle) + 2\gamma_t\sigma_1(M_t)\delta_t + \frac{\gamma_t^2}{2}C_f.
    \end{equation}

    Let $h(X^{(t)}):=f(X^{(t)})-f(X^*)$, we have
    \begin{align}
        h(X^{(t+1)}) & \le h(X^{(t)})+\gamma_t (\min_{\hat{S} \in \cD}\langle{\hat{S}, M_t\rangle}- \langle{X^{(t)}, M_t)}\rangle) + 2\gamma_t\sigma_1(M_t)\delta_t + \frac{\gamma_t^2}{2}C_f \nonumber \\ 
        & \le h(X^{(t)})-\gamma_t h(X^{(t)}) +2\gamma_t\sigma_1(M_t)\delta_t + \frac{\gamma_t^2}{2}C_f \nonumber \\ & =  (1-\gamma_t) h(X^{(t)}) +2\gamma_t\sigma_1(M_t)\delta_t + \frac{\gamma_t^2}{2}C_f.
    \end{align}

    Set $\gamma_{t}=\frac{2}{t+2},\delta_t=\frac{\gamma_t C_f}{4\sigma_1(M_t)}$, we have

    \begin{equation}
        h(X^{(t+1)}) \le  \left(1-\frac{2}{t+2}\right) h(X^{(t)}) + \left(\frac{2}{t+2}\right)^2C_f.
    \end{equation}

    Using a similar induction as shown in \cite{jaggi2013revisiting} over $t$, we have
    \begin{equation}
        h(x^{(t)}) \le  \frac{4C_f}{t+2}.
    \end{equation}

    In summary, set $\gamma_{t}=\frac{2}{t+2},\delta_t=\frac{C_f}{2(t+2)\sigma_1(M_t)}$, after $T=\frac{4C_f}{\varepsilon}-2$ rounds, we have
    \begin{equation}
        f(\bx^{(T)})-f(x^*) \le  \varepsilon,
    \end{equation}
    for any $\varepsilon>0$.
    Since $\delta_t=\frac{C_f}{2(t+2)\sigma_1(M_t)}\ge \frac{C_f}{2(T+2)\sigma_1(M_t)}=\frac{\varepsilon}{2\sigma_1(M)}$, in each round, the time complexity of update computing is $O\left( \frac{\norm{M}_F \sigma_1^2(M) d \cdot \text{poly}\log{d}}{\sqrt{p} (\sigma_1(M)-\sigma_2(M)) \varepsilon^2} \right)$. Since $\norm{M}_F\le\sqrt{r}\sigma_1(M), p \ge \frac{1}{r}$, the time complexity is upper bounded by $O\left( \frac{r \sigma_1^3(M) d \cdot \text{poly}\log{d}}{(\sigma_1(M)-\sigma_2(M)) \varepsilon^2} \right)$, where $r$ is the rank of the gradient matrix.
\end{proof}

\begin{lemma} 
    \label{lemma:technical1}
    For any $\norm{\bx-\bx'}_2,\norm{\by-\by'}_2\le \delta<1$, $\norm{\bx},\norm{\by} \le 1$, we have
    \begin{equation}
        \abs{\bx^{\top}M\by-{\bx'}^{\top}M\by'}\le 2\sigma_1(M)\delta.
    \end{equation}
\end{lemma}

\begin{proof}
    Since $\bx^{\top}M\by-{\bx'}^{\top}M\by'=(x-x')^{\top}M\by+{\bx'}M(\by-\by')$, we have
    \begin{equation}
        \abs{\bx^{\top}M\by-{\bx'}^{\top}M\by'}\le \sigma_1(M)\norm{\bx-\bx'}_2\norm{\by}_2 + \sigma_1(M)\norm{\bx'}_2\norm{\by-\by'}_2.
    \end{equation}
    Thus, for any $\norm{\bx-\bx'}_2,\norm{\by-\by'}_2\le \delta<1$, $\norm{\bx},\norm{\by} \le 1$, we have
    \begin{equation}
        \abs{\bx^{\top}M\by-{\bx'}^{\top}M\by'}\le 2\sigma_1(M)\delta.
    \end{equation}
\end{proof}

\subsection{Proof of Lemma \ref{lemma:qpm}}
\label{proof:lemma:qpm}
\lemmaqpm*
\begin{proof}
Suppose $\norm{z_l-Mz_{l-1}} \le  \epsilon$  with $z_l=Mz_{l-1}$ for $l \in [L]$, and $z_0=x$, we have
\begin{align}
    \norm{z_1-Mx} &\le  \epsilon \nonumber\\
    \norm{z_2-M^2x} &\le \norm{z_2-Mz_1+Mz_1 - M^2x}\nonumber\\
    &\le\norm{z_2-Mz_1} + \norm{Mz_1 - M^2x}\nonumber\\
    &\le \epsilon + \norm{M(z_1-Mx)}\nonumber\\
    &\le \epsilon + \sigma_{\max}\norm{z_1-Mx}\nonumber\\
    &\le (1+\sigma_{\max})\epsilon\nonumber\\
    \norm{z_3-M^3x} &\le \norm{z_3-Mz_2+Mz_2 - M^3x}\nonumber\\
    &\le\norm{z_3-Mz_2} + \norm{Mz_2 - M^3x}\nonumber\\
    &\le \epsilon + \norm{M(z_2-M^2x)}\nonumber\\
    &\le \epsilon + \sigma_{\max}\norm{z_2-M^2x}\nonumber\\
    &\le \epsilon + \sigma_{\max}(1+\sigma_{\max})\epsilon\nonumber\\
    &\le (1+\sigma_{\max}+\sigma_{\max}^2)\epsilon.
\end{align}

We use $\sigma_{\min}\norm{x}\le\norm{Mx} \le \sigma_{\max}\norm{x}$, where $\sigma_{\max}=\max_{x\neq 0}{x^{\top}Mx}/\norm{x}^2$.
By induction, we have 
\begin{equation}
    \norm{z_L-M^Lx}\le \sum_{i \in [L]}\sigma_{\max}^{i-1}\epsilon =\frac{\sigma_{\max}^L-1}{\sigma_{\max}-1}\epsilon.
\end{equation}

Let ${\gamma'}_{\min}$ be the lower bound of $\norm{(M^{\top}M)^i \bz)}$ for all $i\in [k]$. As each multiplication requires time complexity of $\tilde{O}(\frac{1}{\gamma} \norm{M}_F \log(1/\epsilon))$ (Lemma \ref{lem:q_multiply}), $k$ steps of multiplication require time complexity of $\tilde{O}\left(\frac{k}{{\gamma'}_{\min}} \norm{M}_F \log(1/\epsilon)\right)$.
Furthermore, since
\begin{equation}
    \log{\frac{1-\sigma^k_1(M)}{1-\sigma_1(M)}}\le -\log{(1-\sigma_1(M))} \le \frac{1}{1-\sigma_1(M)},
\end{equation}
if we want $\norm{z_k-M^kx}\le\delta$, the time complexity will be $\tilde{O}\left(\frac{k\norm{M}_F}{(1-\sigma_1(M)){\gamma'}_{\min}} \log(1/\delta)\right)$.

\end{proof}

\subsection{Proof of Theorem \ref{theorem:qfwqpm}}
\label{proof:theoremqfwqpm}
\theoremqfwqpm*
\begin{proof}
Denote $(MM^{\top})^k\bb$ as $\bz_u$, $(M^{\top}M)^k\bb$ as $\bz_v$. For the quantum power method, we first use the Lemma \ref{lem:q_multiply} to construct a unitary $U_1$ which computes $k$ steps of multiplication: $U_1:\ket{\bb}\ket{\bb}\rightarrow \ket{\overline{\bz}_u}\ket{\overline{\bz}_v}$ with $\norm{\overline{\bz}_u-\bz_u}_2\le \delta$ and $\norm{\overline{\bz}_v-\bz_v}_2\le \delta$ (Lemma \ref{lemma:qpm}).
Then we tomography $\ket{\overline{\bz}_u}\ket{\overline{\bz}_v}$ to get $\bu,\bv$. 
Simalar to the proof of Theorem \ref{theorem:qfwqsvd}, our goal is to ensure $\abs{\frac{\bu^{\top}M\bv}{\norm{\bu}\norm{\bv}}-\sigma_1(M)}\le \varepsilon$.

First, we settle down $k=\frac{2C_0\sigma_1(M)\ln d}{\varepsilon}$ so that we have
\begin{equation}
    \abs{\frac{\bz_u^{\top}M\bz_v}{\norm{\bz_u}\norm{\bz_v}}-\sigma_1(M)}\le \varepsilon/2.
\end{equation}

Suppose $\sigma_1(M)<1$ and $\norm{(M^{\top}M)^i \bb)}\in [{\gamma'}_{\min},1]$ for $i=1,...,k$, after applying $k$ times of quantum matrix-vector multiplication ($U_1$) as described by \cref{lem:q_multiply}, we obtain $\ket{\overline{\bz}_u}\ket{\overline{\bz}_v}$ with $\norm{\overline{\bz}_u-\bz_u}_2\le \delta$ and $\norm{\overline{\bz}_v-\bz_v}_2\le \delta$ in time $T(U_1)=\tilde{O}\left(\frac{k\norm{M}_F}{(1-\sigma_1(M))\gamma'_{\min}}\log (1/\delta)\right)$. 
Using $U_1$, we can tomography $\ket{\overline{\bz}_u}\ket{\overline{\bz}_v}$ and obtain $\bu,\bv$ with $\norm{\bu-\overline{\bz}_u} \le \delta', \norm{\bv-\overline{\bz}_v} \le \delta'$ in time $O\left(\frac{T(U_1)d\log d}{(\delta')^2}\right)$. 
By the triangle inequality, we have 
\begin{equation}
    \norm{\bu-\bz_u}_2\le \delta+\delta'\le 1, \norm{\bv-\bz_v}_2\le \delta+\delta'\le 1
\end{equation}
Notice that 
\begin{align}
    \norm{\frac{\bu}{\norm{\bu}}-\frac{\bz_u}{\norm{\bz_u}}} & =\norm{\frac{\bu}{\norm{\bu}}-\frac{\bu}{\norm{\bz_u}}+\frac{\bu}{\norm{\bz_u}}-\frac{\bz_u}{\norm{\bz_u}}} \nonumber \\
    & \le \norm{\frac{\bu}{\norm{\bu}}-\frac{\bu}{\norm{\bz_u}}}+\norm{\frac{\bu}{\norm{\bz_u}}-\frac{\bz_u}{\norm{\bz_u}}} \nonumber \\
    & \le 2\frac{\norm{\bu-\bz_u}}{\norm{\bz_u}},
\end{align}
we have 
\begin{equation}
    \norm{\frac{\bu}{\norm{\bu}}-\frac{\bz_u}{\norm{\bz_u}}}\le 2\frac{\delta+\delta'}{{\gamma'}_{\min}}.
\end{equation} 
Similarly, we have
\begin{equation}
    \norm{\frac{\bv}{\norm{\bv}}-\frac{\bz_v}{\norm{\bz_v}}}\le 2\frac{\delta+\delta'}{{\gamma'}_{\min}}.
\end{equation}
Thus, we have 
\begin{align}
    \abs{\frac{\bu^{\top}M\bv}{\norm{\bu}\norm{\bv}}-\frac{\bz_u^{\top}M\bz_v}{\norm{\bz_u}\norm{\bz_v}}} & \le \abs{\frac{\bu^{\top}M\bv}{\norm{\bu}\norm{\bv}}-\frac{\bu^{\top}M\bz_v}{\norm{\bu}\norm{\bz_v}}} + \abs{\frac{\bu^{\top}M\bz_v}{\norm{\bu}\norm{\bz_v}}-\frac{\bz_u^{\top}M\bz_v}{\norm{\bz_u}\norm{\bz_v}}}
    \nonumber \\
    & \le \frac{\norm{M}\norm{\bv-\bz_v}}{\norm{\bz_v}} + \frac{\norm{M}\norm{\bu-\bz_u}}{\norm{\bz_u}} \nonumber \\
    & \le 4\frac{(\delta+\delta')\sigma_1(M)}{{\gamma'}_{\min}}.
\end{align}
The remaining proof is similar to that of Theorem \ref{theorem:qfwqsvd}.
Now we set $\delta=\delta'=\frac{\varepsilon {\gamma'}_{\min}}{16 \sigma_1(M)}$, $\abs{\frac{\bu^{\top}M\bv}{\norm{\bu}\norm{\bv}}-\frac{\bz_u^{\top}M\bz_v}{\norm{\bz_u}\norm{\bz_v}}}\le \varepsilon/2$. Therefore, $\abs{\frac{\bu^{\top}M\bv}{\norm{\bu}\norm{\bv}}-\sigma_1(M)}\le \varepsilon$. The time complexity is $\tilde{O}(T(U_1)d/(\delta')^2)=\tilde{O}(\frac{\sqrt{r}\sigma_1^4(M)d}{(1-\sigma_1(M)){\gamma'}_{\min}^3\varepsilon^3})$, where $r$ is the rank of the gradient matrix.
\end{proof}

\end{document}